\documentclass[aos]{imsart}
\RequirePackage{amsthm,amsfonts,amssymb}
\RequirePackage[numbers,sort]{natbib}
\RequirePackage[colorlinks,citecolor=blue,urlcolor=blue]{hyperref}
\RequirePackage{graphicx}

\usepackage[reqno]{amsmath}
\usepackage{enumerate}
\usepackage{dsfont}
\usepackage[usenames,dvipsnames]{color}
\usepackage{mathrsfs}
\usepackage{placeins}
\usepackage{graphicx}
\usepackage{tikz}
\usetikzlibrary{arrows}
\usepackage{url}
\usepackage{comment}
\usepackage{multicol}
\usepackage{cleveref}
\usepackage{stmaryrd}
\usepackage{soul}
\usepackage{natbib}
\usepackage{algorithm}
\usepackage{algorithmic}

\renewcommand{\thetable}{\arabic{section}.\arabic{table}}
\numberwithin{equation}{section}
\theoremstyle{plain} 
\newtheorem{theorem}{Theorem}[section]
\newtheorem{lemma}[theorem]{Lemma}
\newtheorem{proposition}[theorem]{Proposition}
\newtheorem{corollary}[theorem]{Corollary}

\theoremstyle{definition}
\newtheorem{definition}[theorem]{Definition}

\newtheorem{example}[theorem]{Example}

\newtheorem{remark}{Remark}

\newcommand{\E}{\mathbb{E}}

\renewcommand\P{\mathbb{P}}
\newcommand{\R}{\mathbb{R}}
\newcommand{\bG}{\mathbb{G}}

\def\bI{\mathbb I}

\def\bx{\boldsymbol x}

\def\Var{\text{Var}}

\newcommand{\stp}{\stackrel{\mathbb{P}}{\rightarrow}}

\newcommand{\ov}{\overline}

\newcommand{\vague}{\stackrel{\lower0.2ex\hbox{$\scriptscriptstyle
                    \it{v} $}}{\rightarrow}}
\newcommand{\weak}{\stackrel{\lower0.2ex\hbox{$\scriptscriptstyle
                    \it{w} $}}{\rightarrow}}
\newcommand{\what}{\stackrel{\lower0.2ex\hbox{$\scriptscriptstyle
                    \it{\hat{w}} $}}{\rightarrow}}
\newcommand{\eqdis}{\stackrel{\lower0.2ex\hbox{$\scriptscriptstyle
                    \mathrm{d}$}}{=}}
\newcommand{\distr}{\stackrel{\lower0.2ex\hbox{$\scriptscriptstyle
                    \it{d} $}}{\rightarrow}}

\usepackage{mathtools}
\DeclarePairedDelimiter\ceil{\lceil}{\rceil}

\allowdisplaybreaks 
\definecolor{darkgreen}{RGB}{0,139,0}

\arxiv{math.PR/0000000}

\begin{document}

\begin{frontmatter}
\title{Measuring Extreme Tail Association}
\runtitle{Extreme tail association}

\begin{aug}
  \author{\fnms{Bikramjit} \snm{Das}\ead[label=e1]{bikram@sutd.edu.sg}\orcid{0000-0002-6172-8228}}
    \and
   \author{\fnms{Xiangyu} \snm{Liu}\ead[label=e2]%
   {xliu1932@usc.edu}}
 \address{Engineering Systems and Design, Singapore University of Technology and Design \printead[presep={,\ }]{e1}}
\address{Viterbi School of Engineering, University of Southern California \printead[presep={,\ }]{e2}}


  \runauthor{B. Das and X. Liu}
\end{aug}

\begin{abstract}

Simultaneous occurrences of extreme events need not imply symmetric or reciprocal tail dependence. However, most existing measures of extremal dependence are inherently symmetric and hence often fail to capture directional influence in tail association. We introduce a rank‑based measure of Extreme Tail Association (ETA) for bivariate data quantifying such directional influence of one variable on another in extreme tail regions. The proposed estimator is easily computable, consistent with its population counterpart, and asymptotically normal under mild conditions, allowing for statistical inference. We further develop a formal test for asymmetry in tail association based on a multiplier bootstrap procedure. The practical relevance of the methodology is illustrated using data on extreme price movements in major cryptocurrencies. Beyond providing a flexible tool for extremal association, the proposed framework offers a substantive argument for investigating causal relationships in extreme scenarios.

\end{abstract}

\begin{keyword}[class=AMS]
\kwd[Primary ]{62G32} 
\kwd[; secondary ]{62G10} 
\kwd{60F05} 
\kwd{62G30} 
\end{keyword}

\begin{keyword}
\kwd{Asymmetric distributions}
\kwd{Extremal dependence}
\kwd{Tail association}
\kwd{Testing tail asymmetry}
\end{keyword}

\end{frontmatter}

\section{Introduction}\label{sec:intro}

In extreme value analysis, one key concern has been the proper attribution of  directional association between variables that exhibit extreme values.  For example, extreme weather events such as heat waves, floods, and droughts occur at the same time across different geographic regions; extreme fluctuations in the values of multiple cryptocurrencies have been observed to occur together. Despite occurring simultaneously, there is no reason to believe that the extremal influence between these events is reciprocal in nature; in other words, it is expected that one might have a higher influence on the behavior of the other than the reverse.

An asymmetry in the influence between risk factors for their extreme values has been empirically observed in various domains; for example, \citet{aloui_etal:2011,jondeau:2016} observes an asymmetry in the behavior of extreme movement of US equity portfolios, \citet{kim_etal:2021,gong:huser:2022} have investigated this for high fluctuations in cryptocurrencies,  \citet{deidda_etal:2023,gnecco_etal:2021} have reported on it in a hydrological context for upstream and downstream river flow data. Naturally, certain measures of dependence have also been proposed, see \cite{deidda_etal:2023,mhalla_etal:2020,gnecco_etal:2021, bodik_etal:2023}, mostly in the context of various applications ranging from finance to river networks; yet, few address the asymptotic behavior of such estimates leading to statistical inference on directional tail dependence. 

Traditionally, various measures have also been proposed to capture extremal tail dependence, for example, the  tail dependence coefficient  \cite{sibuya:1960}, or the extremogram \cite{davis:mikosch:2009} or the extremal dependence measure \cite{resnick:2004} have been effective in both quantifying and detecting dependence in the tails; see \cite{resnick:2007, mcneil:frey:embrechts:2005} for further details. Nevertheless, these measures may require further structural assumptions on the data in one or more of the following aspects: (i) they may require the marginal variables to have certain finite moments, (ii) they often depend only on a limit association across a prime diagonal (elaborated later), and  (iii) they are more likely to be symmetric, i.e., exchanging the variables does not change the value of the measure. In this paper, we propose a rank-based measure of tail association, which we call the \emph{Extreme Tail Association} (ETA) coefficient,  which not only overcomes these issues with the traditional measures mentioned above, but is easily computable in sample data and  can be interpreted as an asymmetric complement to correlation measures for measuring tail dependence between variables.  

\subsection{Extreme Tail Association Coefficient (sample)}\label{subsec:defETA} 
Let $(X_{1},Y_{1}), \ldots, (X_{n}, Y_{n})$ be $\R^2$-valued i.i.d.  random vectors with the same distribution as $(X,Y)\sim F$ where the marginal variables $X\sim F_X, Y\sim F_Y$.  For convenience we assume no ties between the $X_i$'s as well as between the $Y_i$'s;  this can be guaranteed almost surely by assuming that $F_X$ and $F_Y$ are continuous, which we do for this paper.

 Denote by $Y_{(1:n)} \ge \ldots \ge Y_{(n:n)}$ the decreasing order statistics of $Y_1,\ldots, Y_n$, and let $X_{[1:n]}, \ldots, X_{[n:n]}$ respectively be their corresponding \emph{concomitant} $X$-variables, meaning that if $Y_{(i:n)}=Y_j$, then $X_{[i:n]}= X_j$. We often omit ``$:n$'' in the subscript for brevity.   Define  the (reverse) rank of the concomitant $X_{[i]}$ among $X_{[1]}, \ldots, X_{[n]}$ to be
 \begin{align}\label{def:revrank}
    \rho_i :=\sum_{j=1}^{n} \bI\left(X_{[j]} \ge X_{[i]}\right).
 \end{align}
where $\bI(A)=1$ if $A$ occurs and $\bI(A)=0$ otherwise, for any event $A$.  Now, for example, $\rho_{\ell}=1$ implies that $X_{[\ell]}$ is the maximum among all $X_i$ values (assuming no ties). 
 \begin{definition}
      Given the above notation, for a fixed integer $k$ with $1 \le k \le n$, define the (sample) upper {\it{Extreme Tail Association (ETA)}} coefficient  of $X$ on $Y$   as
\begin{align}\label{eqdef:eta_data}
\eta_{k,n} (X|Y): =  \frac{3}{k^{3}} \sum_{i=1}^{k-1} \sum_{j=1}^{k-1} (k+1-\max(\rho_{i},\rho_{j}))_{+}.
\end{align}
We often call this the ETA coefficient, or the $\eta$ (coefficient) of $X$ on $Y$, where it is understood that we are concerned with the \emph{upper} tails of both variables in a \emph{sample} context.
 \end{definition}

\subsubsection{Properties of the sample ETA Coefficient} \label{rem:ETAdef}
First, we note some features  of the (sample upper) ETA coefficient.
  \begin{enumerate}[(1)]
\item The ETA coefficient only concerns the upper or right tail of both marginal variables, hence similar quantities can be defined for any combination of lower and upper tail of each marginal distribution; the results for these other combinations of tails are very similar and are not addressed directly here for the sake of brevity.
  \item The measure $\eta_{k,n}$ depends not only on the sample size $n$, but also on the choice of the upper order statistics $k$. This is similar to many statistics used to estimate parameters related to tails or extremes of distributions, e.g., the Hill estimator \cite{hill:1975}, or the Pickands estimator \cite{pickands-75} used to estimate the tail index of a heavy-tailed distribution, or various estimators of the tail dependence coefficient \cite{huang:1992}. Thus, while analyzing such estimators for a given sample size $n$, we seek stability of the estimator in an appropriate range of values for the choice of $k$. One may also explore other choices of $k$; see, for example, \cite{clauset:shalizi:newman:2009,danielsson:dehaan:peng:devries:2001,drees:kauffman:1998} for different proposals in the context of some extreme-value estimators.

\item The value of $\eta_{k,n}(X|Y)$ lies within the range $[0,\left(1-\frac1k\right)\left(1+\frac{5}{2k}-\frac3{k^2}\right)]$. The minimum value, zero, is attained if  $\rho_{i}\ge k+1$ for all $1\le  i < k$, which implies that none of the high values of $Y$ corresponds to a high value in their concomitant $X$'s, naturally signifying that none of the extreme values of $Y$ are associated with the extreme values of $X$. Similarly, the maximum value of $\eta_{k,n} (X|Y)$ is obtained when $\rho_{i}=i$ for all $1\le i < k$, i.e., when the ranks of $X_i$ and $Y_{i}$ match up and  $X_{[i]}$, which is the concomitant  of the $i$-th largest $Y$-value, also happens to be the $i$-the largest $X$-value, and we obtain
\begin{align}\label{eq:ETARHS}
\eta_{k,n} (X|Y) 
&=   \frac{3}{k^{3}}\sum_{i=1}^{k-1} \sum_{j=1}^{k-1} (k+1-\max(i,j))_{+} = \left(1-\frac1k\right)\left(1+\frac{5}{2k} - \frac{3}{k^2}\right),
\end{align}
where the right hand side  tends to 1 as $k\to \infty$. Note that  normalization of $\eta_{k,n}$ by this quantity makes it one, but since we intend to characterize asymptotic tail behavior, we have refrained from doing this. Clearly, a value of $\eta_{k,n}(X|Y)$ close to zero signifies that the extreme values of $X$ do not appear when we observe extreme values of $Y$, i.e., a limited influence of extremes values of $X$ on extremes values of $Y$. On the other hand, a value of $\eta_{k,n}(X|Y)$ close to 1 signifies a higher degree of influence of the extreme values of $X$ on the extreme values of $Y$.

\item Since we have assumed that $X$ and $Y$ are both continuous random variables, we obtain unique orderings of the variables and hence the concomitant ranks as well, but this can be relaxed with minor adjustments to the measure $\eta_{k,n}(X|Y)$. Unfortunately, in such a case the proofs of the associated results become more involved, especially to prove asymptotic normality of $\eta_{k,n}$; we have refrained from doing this at present. 
  \item By construction, $\eta_{k,n}(X|Y)$ is not symmetric in $X$ and $Y$, i.e., $\eta_{k,n}(X|Y)$ and $\eta_{k,n}(Y|X)$  need not be equal. This means that extreme values of one variable may have a higher influence on the extreme values of the other variable even though they might appear to be occurring simultaneously; this also leads us to devising a hypothesis test to assess the asymmetry in tail dependence among variables. We surmise that this would help to attribute causal behavior in extreme tails among these variables. 
 \item Since the sample ETA coefficient is rank-based, it remains unchanged under strictly increasing transformations of the marginal variables $X$ and $Y$. We  observe that this turns out to be true for the limit population ETA measure as well.
\end{enumerate}   

The purpose of this paper is to understand this ETA coefficient $\eta_{k,n}(X|Y)$ and how it helps us assess the asymmetry in tail dependence between the two variables $X$ and $Y$. There exist a  few measures which have attempted to do this, notably \citet{deidda_etal:2023} compute an asymmetric version of Kendall's tau for assessing such an asymmetric dependence in the tails, and \citet{gnecco_etal:2021} address this by defining a causal tail coefficient particularly in the context of heavy-tailed data. There has been work on detecting asymmetry between the lower and upper tails in a financial context; see \cite{bormann_etal:2020, gong:huser:2022} for certain measures and tests proposed; our focus differs in that we restrict attention to the upper tail and aim to characterize directional extremal influence between variables. The ETA coefficient defined in this paper has the advantage of being universally applicable with limited assumptions on the marginal distribution, and in the rest of the paper we show that $\eta_{k,n}$ admits nice asymptotic behavior under mild conditions, making it amenable to estimation, confidence interval computation, and hypotheses testing.

The remainder of the paper is organized as follows. In Section \ref{sec:etatailcop}, the main results are given, including the identification of the population ETA measure and  the proof of the consistency of the coefficient $\eta_{k,n}(X|Y)$. We also propose a tail asymmetry measure between two variables $X,Y$ given by $\Delta_{k,n}(X,Y)= \eta_{k,n}(X|Y) - \eta_{k,n}(Y|X)$ and discuss its properties. The population measure $\eta(X|Y)$ and the tail asymmetry measure $\Delta(X,Y)$ (the population counterpart of $\Delta_{k,n}(X,Y)$) are also shown to be  functionals of the so-called {tail copula} (see \eqref{eq:tailcopup}) which has been quite useful in understanding dependence in the tails. In Section \ref{sec:asybeh}, we  prove the asymptotic normality of the measures $\eta_{k,n}(X|Y)$  and $\Delta_{k,n}(X,Y)$. This helps us to find confidence bounds for $\eta_{k,n}(X|Y)$  and, moreover, provides us with tools for testing tail symmetry in bivariate data. In Section \ref{sec:testasy}, we devise a hypothesis test to assess whether bivariate data exhibit an asymmetric tail dependence behavior with an obvious choice of test statistic given by $\Delta_{k,n}(X,Y)$. Although theoretically justified, finding confidence intervals in real data for such statistics is challenging since the population tail dependence behavior remains unknown (both for null and alternative hypotheses), hence we resort to a method of multiplier bootstrap to obtain appropriate confidence bounds. The method and limit results to do so for both the ETA coefficient  $\eta_{k,n}$  and the test statistic $\Delta_{k,n}$,  are discussed in Section \ref{sec:testasy}. We exhibit our results using data from cryptocurrency movements during the period 2021-2025 in Section \ref{sec:data}; a clear observation here is that the extreme behavior of  positive returns for Bitcoin (BTC) often influences the extreme positive returns of a few other highly traded cryptocurrencies more than the other way around. Section \ref{sec:conclusion} contains some concluding remarks with indications of future directions of work. Proofs of results and additional plots from data analysis are provided in the Appendix.

\section{Measuring association and asymmetry for extreme values}\label{sec:etatailcop}
We intend to understand and assess the measure $\eta_{k,n}$ that has been proposed as a measure of tail association. In this section, we  identify the population measure $\eta(X|Y)$ that the sample measure $\eta_{k,n}(X|Y)$ estimates. We observe that the measure $\eta(X|Y)$ is quite different from $\chi(X,Y)$, the popular measure of tail dependence (cf. \eqref{def:taildep}) and we discuss the relationship between the two. In the process, we also discuss a measure of tail asymmetry $\Delta(X,Y)$ with its sample counterpart $\Delta_{k,n} (X,Y)$.


\subsection{Extreme Tail Association coefficient}\label{subsec:popETA}
The population ETA coefficient defined here, unlike the sample version, is not subject to the assumption of continuity of the marginal variables.
\begin{definition}\label{def:eta_pop}
    Let $(X,Y)\in \R^2$ be a random vector with $X\sim F_X$ and $Y\sim F_Y$.  The (population) \emph{Extreme Tail Association (ETA or $\eta$)} measure of $X$ on $Y$ is given  by
\begin{align}\label{eqdef:eta_pop}
    \eta(X|Y) = \lim_{t\uparrow 1} \eta_t(X|Y) :=  \lim_{t\uparrow 1} \frac{\int_{t}^1\Var(\E(\bI\left(F_{X}(X)>u\right)|\bI\left(F_{Y}(Y)>t\right))) \, \mathrm du}{\int_{t}^1\Var(\E(\bI\left(F_{Y}(Y)>u\right)|\bI\left(F_{Y}(Y)>t\right))) \, \mathrm du}
\end{align}
if it exists.  
\end{definition}
   The limit measure $\eta(X|Y)$ quantifies the co-variability between $X$ and $Y$ in the upper  tail region, specifically when $Y$ is above a specific threshold which eventually increases to its right end point. 
The measure reflects the co-variability in the occurrence of extreme $X$-values when $Y$ exceeds this threshold, thus capturing the influence of increasingly extreme values of $X$ on extreme values of $Y$. 
   In this measure, we evaluate conditional probabilities of the values being above the threshold rather than the conditional expected values; hence this also ensures that the variances within the integrand of both numerator and denominator are finite and, in fact, less than 1, and thus the integrals are eventually finite as well. Thus, we also do not particularly require the existence of a second moment for any of the variables for $\eta(X|Y)$ to exist, which makes it practically useful, especially if the marginal tails are heavy.
\subsubsection{Properties of the population ETA coefficient}
 Proposing a new coefficient to measure the dependence and association for extremes and tails requires some justification. We start by discussing some properties and characteristics of this coefficient.
    \begin{enumerate}[(1)]
        \item First note that $\eta(X|Y)$ exists if the denominator of $\eta_t(X|Y)$ as defined in \eqref{eqdef:eta_pop} is eventually positive (even though it tends to 0). This happens, for example, if $Y$ is continuous near its right tail; in such a case the denominator can be evaluated exactly for $t$ close to 1:
        \begin{align}\label{eq:DenDeltat}
         \int_t^1 \Var(\E(\bI(F_{Y}(Y)>u)|\bI(F_{Y}(Y)>t))) \, \mathrm du = \frac{t(1-t)^2}{3}.
         \end{align}
        We can verify that $0\le \eta(X|Y)\le 1$ when $\eta(X|Y)$ exists, cf. \Cref{prop:eta_is_intutc}.
          \item Clearly $\eta(X|Y)$ is not symmetric by definition. If both $X$ and $Y$ are jointly absolutely continuous random variables, the denominators in the definition of $\eta_t(X|Y)$ and $\eta_t(Y|X)$ are the same, but the numerators may be different. In Section \ref{subsec:tailcopula}, we illustrate this with a few examples.
          \item The measure $\eta(X|Y)$ can be considered as a tail counterpart of the asymmetric correlation measure defined by \citet{chatterjee:2021} which is given by
          \[\xi(Y,X) :=   \frac{\int_{0}^1\Var({\E\left(\bI\left(X>s\right)|Y\right)}) \, \mathrm ds}{\int_{0}^1\Var\left(\bI\left(Y>s\right)\right) \, \mathrm ds}. \]
         The measure $\xi(Y,X)$ is also related to the regression dependence coefficient introduced in \citep{dette_etal:2013}.  Assuming that $X\sim F_X$ and $Y\sim F_Y$ are marginally continuous, we can write 
           \[\xi(Y,X) :=   \frac{\int_{0}^1\Var({\E\left(\bI\left(F_{X}(X)>u\right)|F_Y(Y)\right)}) \, \mathrm du}{\int_{0}^1\Var({\E\left(\bI\left(F_{Y}(Y)>u\right)|F_Y(Y)\right)}) \, \mathrm du}, \]
           which highlights its similarity to $\eta(X|Y)$ defined in \eqref{eqdef:eta_pop}.
        Although similar in flavor, the $\eta(X|Y)$ measure specifically involves the joint tail behavior of the variables and is in fact a limit quantity; thus methods of estimation and analysis are quite different.
       
        \item The measure $\eta(X|Y)$ can also be computed in terms of the stable tail dependence function \cite{beirlant:goegebeur:teugels:segers:2004}, or, upper tail copula \cite{schmidt:stadtmuller:2006}, see \Cref{subsec:tailcopula}. Thus there is a natural relationship between $\eta(X|Y)$ or $\eta(Y|X)$ and the upper tail dependence coefficient used in extreme value theory to measure dependence in the tails. We also provide brief comparisons with other measures of tail dependence in \Cref{subsec:tailcopula}.
        \item It is easy to check that if $g,h:\R\to\R$ are strictly increasing functions, then 
          $$\eta(g(X)|h(Y)) = \eta(X|Y).$$
          Thus, $\eta(X|Y)$ is invariant under monotone increasing transformations of any of the marginal variables.
    \end{enumerate}

Since we have defined ETA as a measure for asymmetric tail dependence, it is natural to define  a measure of asymmetry in the tail derived from the ETA measure.
\begin{definition}\label{def:delta}
    Let $(X,Y)$ be an $\R^2$-valued random vector with $(X,Y)\sim F$.  We define a \emph{measure of tail asymmetry}  given by
\begin{align}\label{eqdef:delta}
    \Delta(X,Y) :=\eta (X|Y) - \eta (Y|X),
\end{align}
where $\eta$ is as defined in \eqref{eqdef:eta_pop}. Clearly, if $(X_{1},Y_{1}), \ldots, (X_{n}, Y_{n})$ are $\R^2$-valued i.i.d.  random vectors with the same distribution as $(X,Y)$ then the sample measure of tail asymmetry is given by
\begin{align}\label{def:sampledelta}
    \Delta_{k,n}(X,Y) :=\eta_{k,n} (X|Y) - \eta_{k,n} (Y|X).
\end{align}
\end{definition}
The quantity $\Delta(X,Y)$ makes sense as a measure of tail asymmetry only when the variables $X$ and $Y$ exhibit some form of tail dependence, that is, at least one of $\eta(X|Y)$ and $\eta(Y|X)$ is non-zero; note here that this differs from the classical notion of tail dependence discussed further in \Cref{rem:taildep}. 
Hence, we may consider three scenarios for the presence of a tail association in the population setting.
\begin{enumerate}
    \item $\eta(X|Y)=\eta(Y|X)=0$ (and hence $\Delta(X,Y)=0$ as well). We may conclude this to be the absence of any dependence in the tails (asymptotically).
    \item $\eta(X|Y)=\eta(Y|X)>0$ and thus $\Delta(X,Y)=0$. Here, the population distribution has tail dependence, but this is symmetric in nature, i.e. the influence of the variables on each other is comparable.
    \item $\Delta(X,Y)=\eta (X|Y) - \eta (Y|X)\neq 0$. Here, the population distribution exhibits tail dependence (asymptotically), but is asymmetric in nature. 
    In particular, here this means that extreme values of $X$ have a relatively greater (or, respectively, smaller) influence on extreme values of $Y$ than vise versa if $\Delta(X,Y) > 0$ (or, respectively, $\Delta(X,Y) < 0$).
\end{enumerate}
The first two scenarios may occur, for example, if $X,Y$ are exchangeable random variables, but in reality, this may be more of an exception than a  norm. We use the empirical measure $\Delta_{k,n}(X,Y)$ as a statistic to test the asymmetry in tail-association in Section \ref{sec:testasy}.

\subsection{Tail Copula and the ETA measure}\label{subsec:tailcopula}

Since $\eta(X|Y)$ and hence $\Delta(X,Y)$ measure tail dependence via integrals of variance of tail probabilities, there seems to be an inherent connection to copulas  \citep{joe:1997} and tail copulas \citep{schmidt:stadtmuller:2006}. We make this clear  and observe that this also allows us to compute  $\eta(X|Y)$ in many examples. Eventually, the proofs of many of our results are facilitated using this representation of $\eta(X|Y)$ through tail copulas. Moreover, it becomes vastly useful both for estimation of confidence intervals and hypothesis testing.

We begin by discussing some relevant concepts and notation related to copulas, tail copulas, and their empirical estimation, see   \cite{joe:1997,schmidt:stadtmuller:2006} for details. For a bivariate  random vector $(X,Y) \sim F$ with  marginal distributions $F_X, F_Y$, the copula $C:[0,1]^2 \to [0,1]$ and the survival copula $\overline{C}:[0,1]^2 \to [0,1]$ are {functions} such that for any $(x,y)\in \R^2$, we have
\begin{align*}
F(x,y) & =\P(X \le x, Y \le y) =C(F_X(x),  F_Y(y)),\\
\overline{F}(x,y) & =\P(X >x, Y>y) =\overline{C}(\overline{F}_X(x), \overline{F}_Y(y)),
\end{align*}
 where for any univariate distribution function $F_X$ we define $\overline{F}_X(x)=1-F(x), x\in \R$. Clearly, the copula $C$ and survival copula $\overline{C}$ for the same distribution $F$ are related and for $(u,v)\in[0,1]^2$ we have the following.
\begin{align}\label{eq:cop:survcop}
\overline{C}(u,v)=u+v-1+C(1-u,1-v).
\end{align}

For measuring dependence in the upper tails, which is of interest in this paper, a quantity often used is the (upper) tail copula. Recall that the (left-continuous) generalized inverse of a distribution function $H$ and its tail/survival function $\overline{H}=1-H$ are defined by
\begin{align*}
    {H}^{\leftarrow}(u) =\begin{cases} \inf\{x\in \R: H(x)\ge u\}, & \text{ if } 0<u\le 1, \\  \sup\{x\in \R: H(x)=0 \}, & \text{ if } u=0.\end{cases}\\
    \overline{H}^{\leftarrow}(u) =\begin{cases} \sup\{x\in \R: \overline{H}(x)\ge u\}, & \text{ if } 0<u\le 1, \\  \inf\{x\in \R: \overline{H}(x)=0 \}, & \text{ if } u=0.\end{cases}
\end{align*}
Therefore, the \emph{(upper) tail copula} $\Lambda: \overline{\R}_+^2\to\R$  is given by  
\begin{align}
  \Lambda(x,y) &: = \lim_{t\to\infty} t\overline{C}(x/t, y/t) = \lim_{t\to\infty} t\overline{F}(\ov{F}_X^{\leftarrow}(x/t), \ov{F}_Y^{\leftarrow}(y/t)), \label{eq:tailcopup}
\end{align}  
provided that the limit exists where $\overline{\R}_+^2=[0,\infty]^2\setminus\{(\infty,\infty)\}$; see \citet{schmidt:stadtmuller:2006} for further details on \emph{tail copulas}  and certain properties that we will refer to throughout the paper. In a similar manner as \eqref{eq:tailcopup}, tail dependence in the left/lower tails or a combination of tails may be investigated using  similarly defined measures, the results we obtain for the upper tails can be similarly derived with equivalent asymptotics and hence are not elaborated in this paper. 
 
Interestingly, the tail copula is in one-to-one correspondence with the so-called \emph{stable tail dependence function} $\ell(x,y):=x+y-\Lambda(x,y)$ for $(x,y)\in \R_+^2$, if either of them exists; see \cite{huang:1992,drees:huang:1998} for more details on this measure. It is well-known that if $F$ is in the maximum domain of attraction of an extreme-value distribution $G$, then this function $\ell$ (and hence $\Lambda$) completely determines $G$; see \cite{einmahl_etal:2012}. 
 Here we establish a relationship between $\eta(X|Y)$ and $\Lambda(\cdot,\cdot)$, which is helpful in proving \Cref{thm:consistency}.

\begin{proposition}\label{prop:eta_is_intutc}
Suppose $(X,Y)\sim F$ with continuous marginal distributions $F_X$ and $F_Y$, respectively. If $\Lambda$ denotes the tail copula as in \eqref{eq:tailcopup} and $\eta(X|Y)$ exists, then  we have the representation 
    \begin{align}\label{eqdef:etarep}
         \eta(X|Y) = 3\int_0^1 [\Lambda(u,1)]^2\, \mathrm du.
    \end{align} 
    Furthermore, $0\le \eta(X|Y) \le 1$.
\end{proposition}
\begin{proof}
    The proof is given in \Cref{app:proofsec2}.
\end{proof}
\begin{remark}\label{rem:taildep}
Here, a comparison of $\eta(X|Y)$ with the classical notion of tail dependence is imperative. The tail dependence between $X$ and $Y$ is traditionally measured by the \emph{tail dependence coefficient} $\chi(X,Y)$ defined as
\begin{align}\label{def:taildep}
\chi(X,Y) = \lim_{u\to0} \P(X>\overline{F}_X^{\leftarrow}(u)|Y>\overline{F}_Y^{\leftarrow}(u))= \lim_{u\to0}\frac{\overline{C}(u,u)}{u} ,
\end{align}
if it exists; see \cite{sibuya:1960,mcneil:frey:embrechts:2005}.
The presence or absence of tail dependence is identified with $\chi(X,Y)>0$ or $\chi(X,Y)=0$, respectively. It is easy to check that, in fact, if it exists, then 
$\chi(X,Y)=\Lambda(1,1)$. On the other hand, $\eta(X|Y)$ requires integrating the function $\Lambda(u,1)^2$. Thus, $\eta(X|Y)=0$ would mean $\chi(X,Y)=\Lambda(1,1)=0$ (almost surely), but we are unable to infer much about $\eta(X|Y)$ or $\eta(Y|X)$ by just knowing the value of $\chi(X,Y)=\Lambda(1,1)$.
\end{remark}

The characterization in \Cref{prop:eta_is_intutc} also gives us an easy criterion to check whether $\Delta(X,Y)=0$, implying a symmetric tail association as given in the following corollary.
\begin{corollary}\label{cor:deltazero}
    If $\Lambda(u,1)=\Lambda(1,u)$ for all $u\in[0,1]$ then $\Delta(X,Y)=0$.
\end{corollary}
Clearly, if the tail copula is symmetric (or, more broadly, the copula is symmetric), then the conditions of \Cref{cor:deltazero} are satisfied. Nevertheless, neither the condition in Corollary \ref{cor:deltazero} nor copula symmetry is necessary for $\Delta(X,Y)=0$ to hold.

\subsection{Examples} \label{subsec:examples} 
We compute the values of $\eta(X|Y), \eta(Y|X)$ and $\Delta(X,Y)$ in a few parametric examples. Traditionally many  applications  use symmetric copulas; and hence examples of parametric asymmetric copula models are few in number; we note a couple of them in this section along with providing some recipes for constructing such asymmetric copulas. We primarily use the representation of  \Cref{prop:eta_is_intutc} to compute the necessary values. 

\begin{example}\label{ex:nelsen} This example is adapted from a popular asymmetric copula (see \cite[Example 2.2]{nelsen:2007}). Suppose $(X,Y)\sim F$ with continuous marginals $F_X, F_Y$ respectively and copula $C_{\theta}$ with $\theta\in [0,1]$ given by
  \[  C_{\theta}(u,v) = \min(u, \theta v + (1-\theta) (u+v-1)_+), \quad \quad 0\le u,v \le 1.\]

We compute both $\eta(X|Y)$ and $\eta(Y|X)$ using  \Cref{prop:eta_is_intutc}; we first compute the tail copula $\Lambda$ for this copula:
\begin{align*}
    \Lambda_{\theta}(u, v) &= \lim_{t \to \infty} t\overline{C}(\frac{u}{t}, \frac{v}{t}) = u + v + \lim_{t \to \infty} t C\Bigl(1 - \frac{u}{t}, 1-\frac{v}{t}\Bigr) - t \\
    &= u + v + \lim_{t \to \infty} t\Bigl(\min\{ 1 - u/t, \theta(1-v/t)+(1-u/t-v/t) \} - 1\Bigr) \\
    &= \begin{cases}
        v, &\text{if } \theta u \geq v, \\
        \theta u, &\text{if } \theta u \leq v.
    \end{cases}
\end{align*}
Then we can compute 
\begin{align*}
    \eta(X|Y) &= 3\int_0^1 [\Lambda_{\theta}(u, 1)]^2 \mathrm du = 3\theta^2 \int_0^1  u^2 \mathrm du = \theta^2,
\intertext{and,}
    \eta(Y|X) &= 3\int_0^1 [\Lambda_{\theta}(1, v)]^2 \mathrm dv= 3 \left[  \int_0^{\theta} v^2 \mathrm dv + \theta^2 \int_{\theta}^1  \mathrm dv   \right] = 3\theta^2-2\theta^3.
\end{align*}  

Computing the tail asymmetry measure $\Delta(X,Y)$ we have
\begin{align}\label{eq:deltaEx21}
\Delta(X,Y) = \eta(X|Y) - \eta(Y|X) = -2\theta^2(1-\theta). 
\end{align}
Now, $\Delta(X,Y)$ attains its maximum value (in an absolute sense) when $\theta=2/3$ and its value is ${\Delta}(X,Y)=-\frac{8}{27}$. As a comparison to these measures, we compute the tail dependence coefficient, which is a symmetric measure given by
  \begin{align*}
   \chi(X,Y) & = \lim_{t\uparrow 1} \frac{\P(F_X(X)>t,F_Y(Y)>t)}{\P(F_X(X)>t)}\\
                  & = \frac{1-2t +C_{\theta}(t,t)}{1-t}  = \frac{1-2t +(t+(1-\theta)(t-1))}{1-t} =\theta.
   \end{align*}  
For example, let us fix $\theta=\frac{2}{3}$. Then the tail dependence coefficient satisfies
\[
\chi(X,Y)=\frac{2}{3}>0,
\]
which implies that $X$ and $Y$ are tail dependent. We also compute
\[
\eta(X\mid Y)=\frac{4}{9}\approx 0.44,
\qquad
\eta(Y\mid X)=\frac{20}{27}\approx 0.74.
\]
Hence, although both directions exhibit tail association, the extremal association measured by $\eta$ is stronger in the direction of $Y$ on $X$ than in the direction of $X$ on $Y$. In other words, extreme values of $X$ are more strongly associated with concurrent extreme values of $Y$ than the reverse. This suggests that, in the tails, $Y$ exerts a stronger extremal influence on $X$ than $X$ does on $Y$.

\end{example}

As mentioned above, examples of asymmetric copulas are not always readily available, yet a mechanism to create asymmetric copulas is to transform a symmetric copula \cite{joe:1997}. Suppose  $C^*$ is a (symmetric) copula. Then a (possibly) asymmetric copula is given by
\begin{align}\label{def:cop:asym}
    C_{\alpha,\beta}(u,v) = u^{1-\alpha}v^{1-\beta} C^*(u^{\alpha},v^{\beta}), \quad (u,v)\in [0,1]^2
\end{align}
where $0\le \alpha, \beta\le 1$ and $\alpha\neq \beta$. Now, this construction also allows for an easy computation of the tail copula in terms of the tail copula of the original symmetric copula as given by the next result.

\begin{proposition}\label{prop:uptc_symtoasym}
    Let $\Lambda_{C}$ be the upper tail copula for any copula $C$, whenever it exists. Then with a copula $C_{\alpha,\beta}$ as defined in \eqref{def:cop:asym} with a symmetric copula $C^*$ , we have
    \begin{align}\label{res:asy:uptailcop}
        \Lambda_{C_{\alpha,\beta}} (u,v)= \Lambda_{C^*} (\alpha u, \beta v), \quad (u,v)\in[0,1]^2.
    \end{align}
\end{proposition}
\begin{proof}
    The proof is in \Cref{app:proofsec2}.
\end{proof}

\begin{example} \label{ex:gumbel}
Let $(X,Y)\sim F$ be jointly distributed with continuous marginals and copula $C_{\alpha,\beta; \delta}$, where  $\alpha,\beta \in [0,1]$ are fixed constants and we use  the construction \eqref{def:cop:asym} on the Gumbel copula $C_{\delta}(u,v)=\exp\{-[(-\log u)^{\delta} + (-\log v)^{\delta}]^{1/\delta}\}, 0\le u,v\le 1$ with fixed $\delta\in [1,\infty)$ to obtain
\begin{align*}
C_{\alpha,\beta} (u,v) &= u^{1-\alpha}v^{1-\beta} C_{\delta}(u^{\alpha}, v^{\beta}) \\
     & = u^{1-\alpha} v^{1-\beta} \exp\{-[(-\alpha \log u)^{\delta} + (-\beta\log v)^{\delta}]^{1/\delta}\} , \quad 0\le u,v\le 1.
\end{align*}
For the Gumbel copula, we can find the upper tail copula to be 
\[\Lambda_{\delta}(u,v) = u+v-(u^{\delta} + v^{\delta})^{1/\delta}\quad  0\le u,v\le 1.\]
Hence using \Cref{prop:uptc_symtoasym}, we have the upper tail copula of $C_{\alpha,\beta; \delta}$ to be
\[\Lambda_{\alpha,\beta; \delta}(u,v) = \alpha u+\beta v-((\alpha u)^{\delta} + (\beta v)^{\delta})^{1/\delta}\quad  0\le u,v\le 1.\]
Now, we may compute $\eta_{\alpha,\beta;\delta}(X|Y)$ and $\eta_{\alpha,\beta;\delta}(Y|X)$ using \eqref{prop:eta_is_intutc} for different choices of $\alpha, \beta, \delta$. Unfortunately, simple analytical expressions for $\eta_{\alpha,\beta;\delta}(X|Y)$ and $\eta_{\alpha,\beta;\delta}(Y|X)$ are unavailable for many of the choices of $\delta$.  However, by fixing $\delta=2$ we are able to compute  the values of $\eta_{\alpha,\beta;2}(X|Y), \eta(Y|X)_{\alpha,\beta;2}$ and $\Delta_{\alpha,\beta;2}(X,Y)=\eta_{\alpha,\beta;2}(X|Y)-\eta_{\alpha,\beta;2}(Y|X)$ numerically and as well provide an analytical expression for $\alpha,\beta\in [0,1]$. An analytical expression for $\Delta_{\alpha,\beta;2}$ is given below.
\begin{align*}
    \Delta_{\alpha,\beta;2}(X, Y) 
  & =  \frac{3 \alpha^{3} \operatorname{asinh}{\left(\frac{\beta}{\alpha} \right)}}{\beta} - \frac{2 \alpha^{3}}{\beta} - 4 \alpha^{2} + \frac{2 \alpha^{2} \sqrt{\alpha^{2} + \beta^{2}}}{\beta} + \alpha \sqrt{\alpha^{2} + \beta^{2}}  \\
    &\quad \; - \frac{3 \beta^{3} \operatorname{asinh}{\left(\frac{\alpha}{\beta} \right)}}{\alpha} + \frac{2 \beta^{3}}{\alpha} + 4 \beta^{2} - \frac{2 \beta^{2} \sqrt{\alpha^{2} + \beta^{2}}}{\alpha}- \beta \sqrt{\alpha^{2} + \beta^{2}}.
\end{align*}
In \Cref{fig:exgumbel}, a graph plotting of the coefficients $\eta_{\alpha,\beta;2}(X|Y), \eta(Y|X)_{\alpha,\beta;2}$ and $\Delta_{\alpha,\beta;2}(X,Y)$ for different values of $0\le \alpha, \beta\le 1$ is given.

\begin{figure}[h!]
    \centering
   \includegraphics[width=14.5cm]{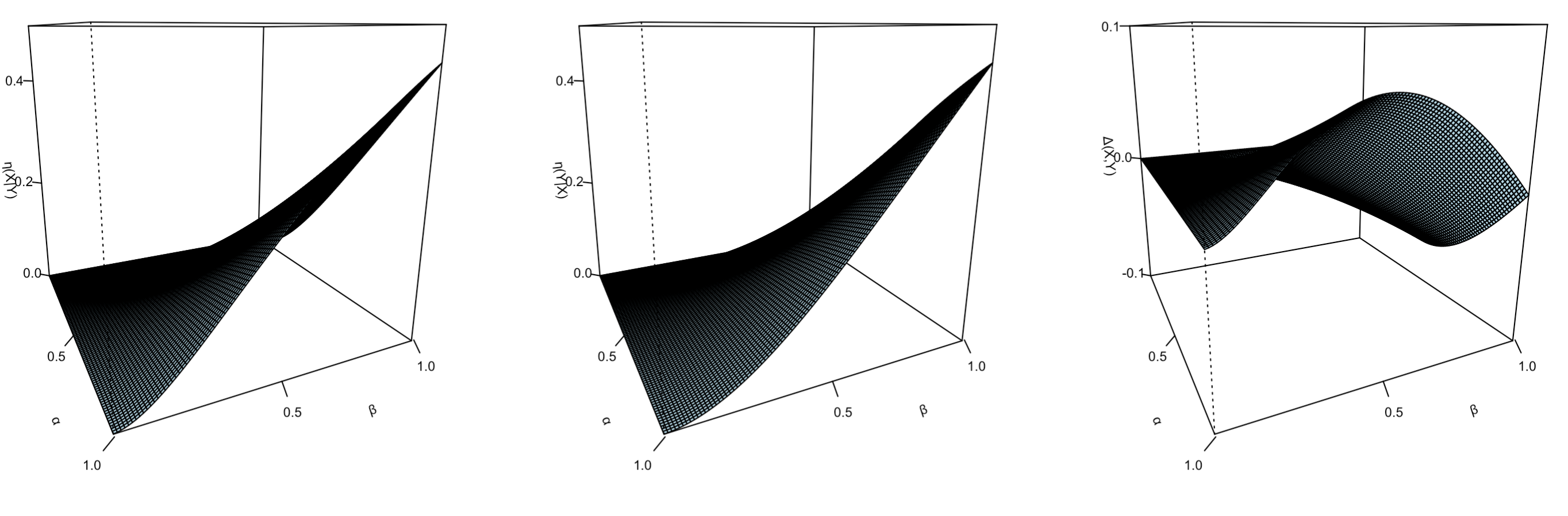}
    \caption{Plot of $\eta(X|Y), \eta(Y|X)$ and $\Delta(X,Y)$ for \Cref{ex:gumbel}.}
    \label{fig:exgumbel}
\end{figure}

The computations also tell us that $\Delta(X,Y)$ attains its  highest (absolute) value when $\alpha=1, \beta=0.5$ (or, $\alpha=0.5, \beta=1$) and we have $\eta_{1,0.5;2}(X|Y)=0.2365, \eta_{1,0.5,2}(Y|X)=0.1685$ and $\Delta_{1,0.5,2}(X,Y)=0.068$. 
\end{example}

\begin{example}\label{ex:maxmodel} 
We consider one more family of examples given by $(X,Y)\sim F$ with continuous marginals $F_X, F_Y$ respectively and copula $C_m$ with $m\in \mathbb{Z}, m \ge 2$,  given by
  \[  C_m(u,v) = \min(u, v^{1/m})\, v^{1-1/m}, \quad \quad 0\le u,v \le 1.\]
Such a distribution may be constructed by considering i.i.d. random variables $Z_1,\ldots,Z_m \sim F_Z$ and defining $X=Z_1$ and $Y=\max(Z_1,\ldots,Z_m)$ and then $(X,Y)\sim F$.
We compute both $\eta(X|Y)$ and $\eta(Y|X)$ using  \Cref{prop:eta_is_intutc}; we first compute the tail copula $\Lambda$ for this copula:
\begin{align*}
    \Lambda_m(u, v) &= \lim_{t \to \infty} t\overline{C}_m(\frac{u}{t}, \frac{v}{t})
     = \min(u,v/m).
\end{align*}
Then we compute 
\begin{align*}
    \eta(X|Y) &= 3\int_0^1 [\Lambda_m(u, 1)]^2 \mathrm du 
    = \frac{3}{m^2}-\frac{2}{m^3},\quad
    \eta(Y|X) = 3\int_0^1 [\Lambda_m(1, v)]^2 \mathrm dv 
    = \frac{1}{m^2},
    \intertext{and,}
    \Delta(X,Y) & = \eta(X|Y) - \eta(Y|X) = \frac2{m^2}\left(1-\frac1m\right).
\end{align*}
In this case, the tail copula takes a similar form as in \Cref{ex:nelsen} with $\theta=1/m$ and the roles of $\eta(X|Y)$ and $\eta(Y|X)$ reversed.
\end{example}
There are different mechanisms to model asymmetry in tails of distributions, but our goal remains to observe this in real data, which we  exhibit in \Cref{sec:data}.

\section{Asymptotic behavior} \label{sec:asybeh}

The empirical quantities for the extremal tail coefficient and the measure of tail asymmetry are easily computable in sample data. Here, we study their asymptotic properties. The characterization of the population ETA coefficient given in \Cref{prop:eta_is_intutc} is helpful in proving its consistency.

\subsection{Consistency of $\eta_{k,n}$ and $\Delta_{k,n}$}\label{subsec:consistency}

First, we demonstrate that  $\eta_{k,n}$ and $\Delta_{k,n}$ are indeed consistent estimators for their population counterparts. The consistency results require only an assumption on the rate of increase of the intermediate sequence $k$, and this is well-known in the study of extremes.

\begin{theorem}\label{thm:consistency}
   Let $(X_1,Y_1), \ldots, (X_n,Y_n)$ be i.i.d from a bivariate distribution $F$ with continuous marginals $F_X, F_Y$.
Then as $n\to \infty, k \to \infty$ and $n/k \to \infty$, 
\begin{align}\label{eq:etaknconvinp}
    {\eta}_{k,n}(X|Y) \to \eta(X|Y) \quad \text{in probability.}
\end{align}

Additionally, if $k/\log\log n \to \infty$ as $n\to \infty$, then
\begin{align}\label{eq:etaknconvas}
    {\eta}_{k,n}(X|Y) \to \eta(X|Y) \quad \text{almost surely.} 
\end{align}

\end{theorem}
\begin{proof}
    The proof is given in \Cref{app:proofsec2}.
\end{proof}

\begin{remark}
    It is not necessary that we assume that both $F_X$ and $F_Y$ are continuous everywhere; it suffices to assume that the marginal distributions are continuous near the right tail.
\end{remark}
The consistency of $\Delta_{k,n}(X,Y)$ is a direct consequence of \Cref{thm:consistency}; we state this result without proof next.
\begin{corollary}\label{cor:delta}
   Suppose $(X_1,Y_1), \ldots, (X_n,Y_n)$ are i.i.d from a bivariate distribution $F$ with continuous marginals $F_X, F_Y$. Then as $n\to \infty, k \to \infty$ and $n/k \to \infty$, 
\[{\Delta}_{k,n}(X,Y) \to \Delta(X,Y) \quad \text{in probability.} \]
Additionally, if $k/\log\log n \to \infty$, as $n\to \infty$, then
 \[{\Delta}_{k,n}(X,Y) \to \Delta(X,Y) \quad \text{almost surely.} \]  
\end{corollary}

\subsection{Asymptotic normality of $\eta_{k,n}$ and $\Delta_{k,n}$}

Now we prove that these estimators are in fact  asymptotically normal under proper scaling. Naturally, this requires further assumptions and some new notation as well.

The representation of $\eta(X|Y)$ in \Cref{prop:eta_is_intutc} clearly suggests a natural estimator using an empirical version of the tail copula $\Lambda$. Let $(X_{1},Y_{1}), \ldots,$  $(X_{n}, Y_{n})$ be i.i.d. pairs of random variables with joint distribution $F$ and marginals $F_X$ and $F_Y$. Following \cite{bucher:dette:2013}, the empirical  distribution and survival distribution function  of $F$ are respectively
\begin{align*}
& F_n(x,y) = \frac1n\sum_{i=1}^n \bI(X_i\le x, Y_i \le y),  \text{ and, } \,\, \ov{F}_n(x,y) = \frac1n\sum_{i=1}^n \bI(X_i> x, Y_i> y), \end{align*}
for $(x,y)\in \R^2$. Similarly, the empirical marginal distribution and survival distribution of $F_X, F_Y$ are
\begin{align*}
& F_{n,X}(x) = F_n(x,\infty), \;\;\; \ov F_{n,X}(x) = 1- F_{n,X}(x),\\
 & F_{n,Y}(y) = F_n(\infty,y), \;\;\; \ov F_{n,Y}(y) = 1- F_{n,Y}(y).
\end{align*}
Following \cite{schmidt:stadtmuller:2006, bucher:dette:2013}, the upper tail copula can be estimated consistently by the empirical tail copula given by
\begin{align}\label{eq:lambdahat}
\widehat{\Lambda}_{k,n}(x,y) = \frac{n}{k} \ov{C}_n\left(kx/n,ky/n\right), \quad\quad (x,y) \in\R^2,
\end{align}
where
\[\ov{C}_n(u,v) = \ov{F}_n(\ov{F}_{n,X}^{\leftarrow}(u), \ov{F}_{n,Y}^{\leftarrow}(v)), \quad (u,v)\in [0,1]^2. \]

In the literature of \emph{extreme value theory}, it is well-known that asymptotic normality of estimates of tail parameters often requires a further \emph{second-order} condition for the limit distribution (or limit tail copula or stable tail dependence function). We collect a few other conditions on the tail copula  $\Lambda$ below which will facilitate such a proof; see \cite{schmidt:stadtmuller:2006, bucher:dette:2013}.

\begin{itemize}
\item[(C1)] The tail copula $\Lambda(\cdot,\cdot)$ is not identically 0.
\item[(C2)] For some function $A:\R_+\to\R_+$ with $\lim_{t\to\infty} A(t) = 0$, we have 
\begin{align*}
    \lim_{t\to\infty} \frac{\Lambda(x, y) - tC(x/t, y/x)}{A(t)} = g(x, y) 
\end{align*}
locally uniformly for $(x, y) \in \overline{\mathbb{R}}^2_{+}$ for some non-constant function $g$.
\item[(C3)] The partial derivatives 
   $$\frac{\partial}{\partial x} \Lambda(x,y), \quad \text{and}, \quad \frac{\partial}{\partial y} \Lambda(x,y),$$
exist and are continuous for all $x\in(0,\infty)$ and $y\in(0,\infty)$, respectively.
\item[(C4)] As $k\to\infty, n\to\infty$ we have $\sqrt{k}A(n/k) \to 0$ as $n \to \infty$.
\end{itemize}
First we state a result on the weak convergence of the tail copula process; the result is adapted from \cite{bucher:dette:2013} where the result is stated for the lower tail copula, we have adapted this for the upper tail copula.
\begin{proposition}\label{thm:bucherdette}(\citet[Theorem 2.2]{bucher:dette:2013})
  Let $(X_1,Y_1), \ldots, (X_n,Y_n)$ be i.i.d. from a bivariate distribution $F$ with continuous marginals $F_X, F_Y$. Let  conditions (C1)-(C4) hold for the associated tail copula $\Lambda$. Then, as $n\to \infty, k\to \infty$ and $n/k \to \infty$, we have
  \begin{align*}
   \sqrt{k}\left(\widehat{\Lambda}_{k,n}(x,y) - \Lambda(x,y) \right) \weak \bG_{\widehat{\Lambda}}(x,y)
  \end{align*}
  where $\weak$ here denotes weak convergence in sup-norm over each
compact set in $\overline{\R}^2\setminus\{(0,0)\}$  and $\bG_{\widehat{\Lambda}}$ is a bivariate Gaussian field with
  \[\bG_{\widehat{\Lambda}} (x,y) =  \bG_{\tilde{\Lambda}} (x,y) -\frac{\partial \Lambda(x,y)}{\partial x} \bG_{\tilde{\Lambda}} (x,\infty)-\frac{\partial \Lambda(x,y)}{\partial y} \bG_{\tilde{\Lambda}} (\infty,y)\]
  with $\bG_{\tilde{\Lambda}}$ a centered bivariate Gaussian field with covariance  given by
  \[ \E\left(\bG_{\tilde{\Lambda}}(x,y) \bG_{\tilde{\Lambda}}(\tilde{x},\tilde{y})  \right) =  \Lambda (x\wedge \tilde{x}, y \wedge \tilde{y})\]
for $(x,y), (\tilde{x},\tilde{y}) \in \R_+^2$.
\end{proposition}

  Now the above allows us to obtain asymptotic normality for the measure $\eta_{k,n}(X|Y)$ which we state below.


\begin{theorem}\label{thm:asymnorm}
   Let the assumptions of \Cref{thm:consistency} hold, and additionally we assume that conditions (C1)-(C4) hold as well. Then, as $n\to \infty, k\to \infty$ and $n/k \to \infty$, we have
   $$ \sqrt{k}\big(\eta_{k,n}(X|Y) - \eta(X|Y)\big) \Rightarrow 6\int_0^1 \bG_{\widehat{\Lambda}}(u,1) \Lambda(u,1) \; \mathrm du,$$
  where $G_{\tilde{\Lambda}}$ is the centered bivariate Gaussian field from \Cref{thm:bucherdette}. Hence $\eta_{k,n}(X|Y)$ is asymptotically normal with mean $\eta(X|Y)$ and variance given by $\sigma_{\eta(X|Y)}^2/k$ where 
  \begin{align}\label{eq:vartildeLambda}
  \sigma_{\eta(X|Y)}^2= 36\int\limits_0^1\int\limits_0^1 \Lambda(u,1)\Lambda(v,1) \E[G_{\widehat{\Lambda}}(u,1)G_{\widehat{\Lambda}}(v,1)]\, \mathrm du \, \mathrm dv.
  \end{align}
\end{theorem}
\begin{proof}
    The proof is given in \Cref{app:proofsec3}.
\end{proof}

The following proposition gives the limit distribution of $\Delta_{k,n}(X,Y)$.

\begin{theorem}\label{thm:teststatdist}
   Let the assumptions of \Cref{thm:consistency} hold, additionally we assume that conditions (C1)-(C4) hold as well. Then, as $n\to \infty, k\to \infty$ and $n/k \to \infty$, we have
     $$\sqrt{k}\big(\Delta_{k,n}(X,Y) - \Delta(X,Y)\big) \Rightarrow  6\int_0^1 \left[\bG_{\widehat{\Lambda}}(u,1) \Lambda(u,1) - \bG_{\widehat{\Lambda}}(1,u) \Lambda(1,u)\right] \; \mathrm du,$$
   where $G_{\tilde{\Lambda}}$ is the centered bivariate Gaussian field from \Cref{thm:bucherdette}. 
\end{theorem}
\begin{proof}
    The proof is given in \Cref{app:proofsec3}.
\end{proof}
Clearly, from \Cref{thm:teststatdist}, we have that $\sqrt{k}(\Delta_{k,n}(X,Y)-\Delta(X,Y))$ is asymptotically normal with mean $0$ and variance $\sigma^2_{\Delta}$ which in general can be only be computed if we know the tail copula $\Lambda(X,Y)$. Nevertheless, when we have real data, the copula and hence the tail copula remain unknown; thus a well-established technique to  estimate  ${\sigma}_{\Delta}^2$ is by using bootstrapping, in particular a version of the multiplier bootstrap method; see \cite{bucher:dette:2013}.

\section{Testing for tail asymmetry}\label{sec:testasy}
We have defined the measure of tail asymmetry as follows
\[\Delta(X,Y) = \eta(X|Y)-\eta(Y|X),\]
and, naturally, $\Delta(X,Y)=0$ indicates symmetry in tail dependence between the variables if at least $\eta(X|Y)$ or $\eta(Y|X)$ is positive. Hence, we propose the following test of tail symmetry/asymmetry between $X$ and $Y$:
\begin{align}
    \mathcal{H}_0: \Delta(X,Y)=0 \quad\text{ vs. } \quad \mathcal{H}_1: \Delta(X,Y)\neq 0.\label{eq:hyptestasym}
\end{align}
We use the empirical estimator for $\Delta(X,Y)$ as a test statistic, given by
\[\Delta_{k,n}(X,Y) = \eta_{k,n}(X|Y)-\eta_{k,n}(Y|X).\]
Note that one may also test for $\eta(X|Y)=0$ and $\eta(Y|X)=0$ and perform the test in \eqref{eq:hyptestasym} only if tail-dependence is expected to be present.

In Theorems \ref{thm:asymnorm} and \ref{thm:teststatdist}, we have characterized the limit behavior of the test statistic $\Delta_{k,n}$ (and $\eta_{k,n}(X|Y), \eta_{k,n}(Y|X)$) which turns out to be asymptotically normal with a variance that is a function of the tail copula $\Lambda$. The tail copula is naturally unknown in real data, and hence we resort to a multiplier bootstrap method to estimate the variance of the limiting normal distribution following \citet{bucher:dette:2013}.
\subsection{Asymptotic behavior of the (multiplier) bootstrap estimator}\label{subsec:multbootasym}
First, we recall the multiplier bootstrap method used in \citet{bucher:dette:2013} to estimate the tail copula and then use it to propose bootstrap estimators for $\Delta(X,Y)$ and $\eta(X|Y)$. 

Let $\{\xi_i\}_{1\le i \le n}$ be independent and  identically distributed positive random variables with finite mean $\mu$ and finite variance $\tau^2$ and independent of $\{(X_i, Y_i)\}_{1\le i \le n}$. Let $\overline{\xi}_n := \frac1n \sum_{i=1}^n \xi_i$. Then we may construct the following modified empirical joint, marginal, and copula distributions:
\begin{align*}
    \ov{F^{\xi}_n}(x, y) &= \frac{1}{n} \sum_{i = 1}^{n} \frac{\xi_i}{\overline{\xi}_n} \bI(X_i > x, Y_i > y), \\
    \ov{F^{\xi}_{n, 1}}(x) &= \frac{1}{n} \sum_{i = 1}^{n} \frac{\xi_i}{\overline{\xi}_n} \bI(X_i > x), \quad
    \ov{F^{\xi}_{n, 2}}(y) = \frac{1}{n} \sum_{i = 1}^{n} \frac{\xi_i}{\overline{\xi}_n} \bI(Y_i > y), \\
    \ov{C^{\xi, \xi}_n}(x, y) &= \ov{F^{\xi}_n}((\ov{F^{\xi}_{n, 1}})^{\leftarrow}(x), (\ov{F^{\xi}_{n, 2}})^{\leftarrow}(y)). 
\end{align*}
Now, we can define a bootstrap estimate of the tail copula $\Delta(X,Y)$ as follows:
\begin{align}\label{def:empboottailcop}
    \Lambda^{\xi, \xi}_{k, n}(x, y) := \frac{n}{k}\ov{C^{\xi, \xi}_n}\left(\frac{kx}{n}, \frac{ky}{n}\right), \quad x\ge0, y\ge 0.
\end{align}
This allows us to define (multiplier) bootstrap estimates for $\Delta(X,Y)$ and $\eta(X|Y)$ for the random sample $\{(X_i, Y_i)\}_{1\le i \le n}$ with a randomly generated i.i.d sequence $\{\xi_i\}_{1\le i \le n}$ with finite $\E(\xi_1) = \mu$ and $\text{Var} (\xi_i)=\tau^2<\infty$ as follows:
\begin{align}
    \eta^{\xi, \xi}_{k, n} (X|Y) &:= 3 \int_0^1 [\Lambda^{\xi, \xi}_{k, n}(x, 1)]^2 \, \mathrm{d}x, \label{def:booteta}\\
    \Delta^{\xi, \xi}_{k, n} (X, Y) &:= \eta^{\xi, \xi}_{k, n} (X|Y) - \eta^{\xi, \xi}_{k, n} (Y|X). \label{def:bootdelta}
\end{align}

The following is a reformulation of \citet[Theorem 3.4 ]{bucher:dette:2013} which we state below (without proof). This is useful to justify the use of bootstrap estimates for $\eta(X|Y)$ and $\Delta(X,Y)$.
\begin{proposition}\label{prop:bootstrappingEmpiricalProcesses}[\citet[Theorem 3.4]{bucher:dette:2013}] 
    Let the assumptions, conditions ((C1)-(C4)) and the notation of \Cref{prop:uptc_symtoasym} hold. Then, with the multiplier bootstrap estimator of the tail copula as defined in \eqref{def:empboottailcop}, we have, as $n\to\infty, k\to\infty$ with $n/k \to\infty$:
    \begin{align*}
        \alpha^{dm}_{n, k} := \frac{\mu}{\tau} \sqrt{k} (\Lambda^{\xi, \xi}_{k, n} - \widehat{\Lambda}_{k, n}) \weak_\xi \mathbb{G}_{\widehat{\Lambda}} \qquad \text{in $\mathcal{B}_{\infty}(\overline{\mathbb{R}}^2_+)$}.
    \end{align*}
\end{proposition} 
The notation $\weak_\xi$ denotes conditional weak convergence for any function  $$H_n=H_n(((X_1,Y_1), \ldots, (X_n,Y_n)),(\xi_1, \ldots, \xi_n))$$ given the data $\{(X_i,Y_i)\}_{1\le i \le n}$ as defined in \cite{Kosorok2008}; see \Cref{app:proofsec4} for more details. Finally, we state the main result which allows us to conduct the tests discussed previously using the multiplier bootstrap method.

\begin{theorem} \label{thm:bootstrappingCLT}
    Under the assumptions, notation, and conditions of Proposition \ref{prop:bootstrappingEmpiricalProcesses}, we have, as $n\to\infty, k\to\infty$ with $n/k \to\infty$:
    \begin{align}
        \frac{\mu}{\tau} \sqrt{k} (\eta^{\xi, \xi}_{k, n}(X|Y) - \widehat{\eta}_{k, n}(X|Y)) &\weak_\xi 6\int\limits_0^1 \Lambda(u, 1) \mathbb{G}_{\widehat{\Lambda}}(u, 1) \, \mathrm{d}u, \label{eq:bootasynormeta}\\
        \frac{\mu}{\tau} \sqrt{k} (\Delta^{\xi, \xi}_{k, n}(X, Y) - \Delta_{k, n}(X, Y)) &\weak_\xi 6 \int\limits_0^1 \left[\Lambda(u, 1)\mathbb{G}_{\widehat{\Lambda}}(u, 1) - \Lambda(1, u)\mathbb{G}_{\widehat{\Lambda}}(1, u)\right] \mathrm{d}u.  \label{eq:bootasynormdelta}
    \end{align}
\end{theorem}

Of course, the result tells us that the bootstrap estimators are asymptotically normal with a potentially unknown variance, which are the same as the limit distributions in Theorems \ref{thm:asymnorm} and \ref{thm:teststatdist}).  Fortunately, since we can create multiple bootstrap estimates by generating fresh i.i.d. samples of $\{\xi_i\}_{1\le i \le n}$, we can estimate the variance of the limit normal distribution quite accurately. The algorithms for these tests are discussed hence.

\subsection{Algorithm for the testing procedure}\label{subsec:algo}
We can write the multiplier bootstrap estimators given in \eqref{def:booteta} and \eqref{def:bootdelta} also in terms of bootstrapping ranks which makes the computation of the integrals a little more manageable, see \Cref{app:proofsec4} for details on this. For the rest of this section, we discuss the algorithm to perform the tests discussed in \Cref{subsec:multbootasym}.

Let $(X_1, Y_1), \ldots, (X_n,Y_n)$ be i.i.d. samples for which we are interested in testing if there is any tail association at all and if so, whether the association is symmetric.
This can be done by estimating $\eta(X|Y), \eta(Y|X)$ and $\Delta(X,Y)$ using $\eta_{k,n}(X|Y), \eta_{k,n}(Y|X)$ and $\Delta_{k,n}(X,Y)$, respectively, and then performing the following tests:
 \begin{enumerate}[(1)]
 \item Test for $\mathcal{H}_0:\eta(X|Y)=0$ vs. $\mathcal{H}_1:\eta(X|Y)>0$ and $\mathcal{H}_0:\eta(Y|X)=0$ vs. $\mathcal{H}_1:\eta(Y|X)>0$. If either of the tests results in rejection, then it indicates presence of tail association, and hence we go ahead and perform the following test of asymmetry in tail dependence. Note that since our test statistics depend on $k$, the tests are performed for a set of values of $k$, given by $\mathcal{K}$ and a rejection would mean rejection for a significant proportion of $\mathcal{K}$ (say, 75\% or 80\% and above).  
 \item Test $\mathcal{H}_0:\Delta(X|Y)=0$ vs. $\mathcal{H}_1:\Delta(X|Y)\neq0$. The rejection criterion is similar to the tests above.
 \end{enumerate} 

The asymptotic behavior of the test statistics requires both $n$ and $k$ to be large, as well as $n/k$ to be large. Since we do not have a value of $k$ specified, we perform the test for a set of values of $k$, given by $\mathcal{K}$. The set $\mathcal{K}$ is supposed to contain values of $k$ satisfying the previous conditions, in practice we take values running between 5\% to 20\% of the value of $n$. We reject $\mathcal{H}_0$ if our tests  reject $\mathcal{H}_0$ for  a significant set of values of $k\in \mathcal{K}$. First, we state the algorithm for testing $\mathcal{H}_0:\eta(X|Y)=0$ vs. $\mathcal{H}_1:\eta(X|Y)\neq0$.

\begin{algorithm}
\caption{Hypothesis test for $\mathcal{H}_0:\eta(X|Y)=0$ vs. $\mathcal{H}_1:\eta(X|Y)>0$}\label{algo:eta}
      \begin{algorithmic}
        \STATE Fix a sequence of values for $k$, given by  the set $\mathcal{K}$ (say, within $5\%$ to $20\%$ of the value of $n$).
        \STATE Initialize counts $m_k=0$ for $k\in \mathcal{K}$.
        \STATE Fix number of bootstraps $B$, and finite values of $\mu, \tau^2$. 
        \FOR {$b$ = $1$ to $B$}
         \STATE Generate an i.i.d  sequence $\xi_1^{(b)}, \ldots, \xi_n^{b}$ with $\E(\xi_1) = \mu, \text{Var} (\xi_1) = \tau^2.$    
         \WHILE {$k\in \mathcal{K}$}
               \STATE Compute $\eta_{k,n}^{\xi^{(b)}, \xi^{(b)}} (X|Y)$
                \STATE Update $m_k \leftarrow m_k+1$ if $\eta_{k,n}^{\xi^{(b)}, \xi^{(b)}} (X|Y) - \eta_{k,n} (X|Y) > \eta_{k,n} (X|Y)$ 
         \ENDWHILE
   \ENDFOR
         \STATE Report p-values $p_k=\frac{m_k}B$ for  $k \in \mathcal{K}$
    \end{algorithmic}
\end{algorithm}

For a fixed level $\alpha$ for the tests, the p-values will indicate whether we reject the hypotheses or not. We may also compute $(1-\alpha)100\%$ confidence intervals by  estimating (from the output of \Cref{algo:eta}) a bootstrapped standard deviation for the limit normal distribution in \eqref{eq:bootasynormeta} as

\begin{align}
s_{\eta_{k,n}}^{boot}=\sqrt{\frac1{B-1}\sum_{i=1}^B (\eta_{k,n}^{\xi^{(b)}, \xi^{(b)}} (X|Y)- \overline{\eta_{k,n}^{boot}})^2}  \label{eq:etasdboot}
\end{align}
where  $\overline{\eta_{k,n}^{boot}} = \frac1B\sum_{i=1}^B \eta_{k,n}^{\xi^{(b)}, \xi^{(b)}} (X|Y).$ Then a $(1-\alpha)100\%$ two-sided confidence interval for $\eta(X|Y)$ for any $k\in \mathcal{K}$ is given by
\begin{align}\label{eq:bootetaCI}
\left[ \eta_{k,n} (X|Y) - z_{\alpha/2}\frac{s_{\eta_{k,n}}^{boot}}{\sqrt{k}},  \eta_{k,n} (X|Y) + z_{\alpha/2}\frac{s_{\eta_{k,n}}^{boot}}{\sqrt{k}} \right]
\end{align}
where $z_{\beta}$ is the $(1-\beta)$-quantile of the standard normal distribution, i.e., $\P(Z>z_{\beta})=\beta$  where $Z\sim \mathcal{N}(0,1)$.

Similarly, we can state the algorithm for testing $\mathcal{H}_0:\Delta(X,Y)=0$ vs. $\mathcal{H}_1:\Delta(X,Y)\neq0$ as given below. 

\begin{algorithm}
\caption{Hypothesis test for $\mathcal{H}_0:\Delta(X,Y)=0$ vs. $\mathcal{H}_1:\Delta(X,Y)\neq0$}\label{algo:delta}
      \begin{algorithmic}
        \STATE Fix a sequence of values for $k$, given by  the set $\mathcal{K}$ (say, within $5\%$ to $20\%$ of the value of $n$).
        \STATE Initialize counts $m_k=0$ for $k\in \mathcal{K}$.
        \STATE Fix number of bootstraps $B$, and finite values of $\mu, \tau^2$. 
        \FOR {$b$ = $1$ to $B$}
         \STATE Generate an i.i.d  sequence $\xi_1^{(b)}, \ldots, \xi_n^{b}$ with $\E(\xi_1) = \mu, \text{Var} (\xi_1) = \tau^2.$    
         \WHILE {$k\in \mathcal{K}$}
               \STATE Compute $\Delta_{k,n}^{\xi^{(b)}, \xi^{(b)}} (X,Y)$
                \STATE Update $m_k \leftarrow m_k+1$ if $|\Delta_{k,n}^{\xi^{(b)}, \xi^{(b)}} (X,Y)-{\Delta_{k,n} (X,Y)}| > |\Delta_{k,n} (X,Y)|$
         \ENDWHILE
   \ENDFOR
         \STATE Report p-values $p_k=\frac{m_k}B$ for  $k \in \mathcal{K}$
    \end{algorithmic}
\end{algorithm}
We may also compute $(1-\alpha)100\%$ confidence intervals for $\Delta(X,Y)$ following the same method as that for $\eta(X|Y)$ as:
\begin{align}\label{eq:bootdeltaCI}
\left[ \Delta_{k,n} (X,Y) - z_{\alpha/2}\frac{s_{\Delta_{k,n}}^{boot}}{\sqrt{k}},  \Delta_{k,n} (X,Y) + z_{\alpha/2}\frac{s_{\Delta_{k,n}}^{boot}}{\sqrt{k}} \right]
\end{align}
for any $k\in \mathcal{K}$ where
\begin{align}
s_{\Delta_{k,n}}^{boot}=\sqrt{\frac1{B-1}\sum_{i=1}^B (\Delta_{k,n}^{\xi^{(b)}, \xi^{(b)}} (X,Y)- \overline{\Delta_{k,n}^{boot}})^2}  \label{eq:deltasdboot}
\end{align}
with $\overline{\Delta_{k,n}^{boot}} = \frac1B\sum_{i=1}^B \Delta_{k,n}^{\xi^{(b)}, \xi^{(b)}} (X,Y).$

For convenience, we often fix the values of $\mu=\tau=1$ and the i.i.d. sample of $\{\xi_i^{(b)}\}_{1\le i \le n}$ to be generated from an exponential distribution with mean 1.

\section{Data Analysis} \label{sec:data}

In this section, we investigate finite-sample performance of the (multiplier-bootstrap) testing procedure developed in Section \ref{sec:testasy}. The code implementation is conducted using the Python programming language and the associated codes  are available in \url{https://github.com/LiuXiangyu99/ETA}. 

We have particularly concentrated on testing for $\mathcal{H}_0:\Delta(X,Y)=0$ vs $\mathcal{H}_1:\Delta(X,Y)\neq 0$, i.e., a test of asymmetry in bivariate upper tail association; besides this, we also construct associated confidence bounds as a bi-product, see \Cref{algo:delta}. This can be modified to test $\mathcal{H}_0: \Delta(X,Y) = \Delta_0$ vs $\mathcal{H}_1: \Delta(X,Y) \neq \Delta_0$ for any $\Delta_0\in [-1,1]$.  Similarly, finite-sample tests for testing $\mathcal{H}_0:\eta(X|Y) =\eta_0$ along with confidence bound construction can be implemented for particular  values of $\eta_0\in [0,1]$  by appropriately modifying the codes provided, \Cref{algo:eta} gives this for $\eta_0=0$.

\begin{figure}[t]
    \centering
    \includegraphics[width=0.8\textwidth]{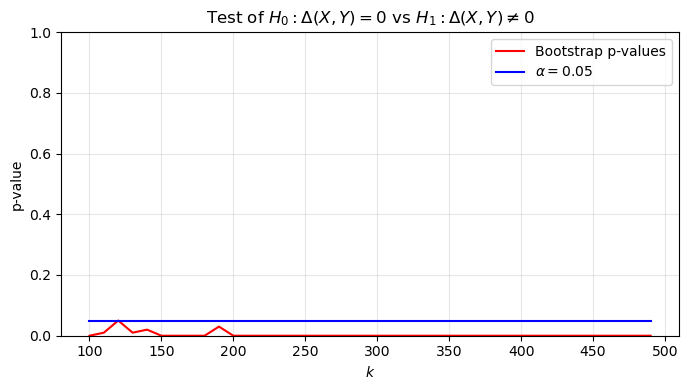}
    \caption{
    Bootstrap p-values using $B=100$ multiplier bootstrap samples for simulated data from \Cref{subsec:datasimul} for testing $\mathcal{H}_0:\Delta(X,Y)=0$ versus $\mathcal{H}_1:\Delta(X,Y)\neq 0$. The red trajectory is p-values across different choice of the intermediate sequence $k\in\mathcal{K}=\{100,110, \ldots, 500\}$. 
    The blue line marks the reference significance level $\alpha=0.05$.
    }
    \label{fig:asymmetrized-gumbel-delta-pvalues}
\end{figure}

\subsection{Simulated data}\label{subsec:datasimul}
First, we implement our testing method on a simulated dataset. Consider the bivariate distribution as discussed in Example \ref{ex:gumbel}, having uniform marginals with asymmetric Gumbel copula $C_{\alpha,\beta;\delta}(\cdot,\cdot)$, where $\alpha,\beta\in [0,1]$ and $\delta\in [1,\infty)$ . Fixing $\delta=2$, we  observe that the maximum value of $\Delta_{\alpha,\beta;2}(X,Y)$ is obtained when  $\alpha=1.0, \beta=0.5$, and in this case $\Delta_{1,0.5;2}(X,Y) \approx 0.068$. Although it is nonzero, the value of  $\Delta_{1,0.5;2}(X,Y)$ is quite close to 0, and thus provides a reasonable scenario for testing the hypothesis that $\mathcal{H}_0:\Delta_{1,0.5;2}(X,Y)=0$ vs. $\mathcal{H}_1:\Delta_{1,0.5;2}(X,Y)\neq0$ .

Here we generate data from this particular bivariate distribution fixing the parameters $\alpha=1, \beta=0.5, \delta=2$. 
We sample $n=5000$ data points from this distribution. A set of values of the intermediate sequence $k \in \mathcal{K}=\{100,110, \ldots, 500\}$ is chosen. In extreme value theory, it is customary to compute the relevant extreme value estimator (statistic) for a few values of $k$ and assess whether they stabilize at some value, ideally the population parameter value $\Delta_{1,0.5;0.2}=0.068$ which is being estimated. A reason for this is that in our asymptotic results, the value of  $k$ also increases along with $n$ without necessarily having a direct functional relationship. We first estimate $\Delta_{k,n}(X,Y)$ and then create $B=100$ multiplier bootstrap samples as discussed in \Cref{sec:testasy}. The p-values are plotted for the various values of $k$ in Figure \ref{fig:asymmetrized-gumbel-delta-pvalues}. The plot clearly indicates that the p-values are less than $0.05$ for almost all of the $k$-values considered, and so we may comfortably reject the null hypothesis that $\Delta=0$.

\subsection{Cryptocurrency returns data}\label{subsec:datacrypto}

\begin{figure}[t]
    \centering
    \includegraphics[width=0.85\textwidth]{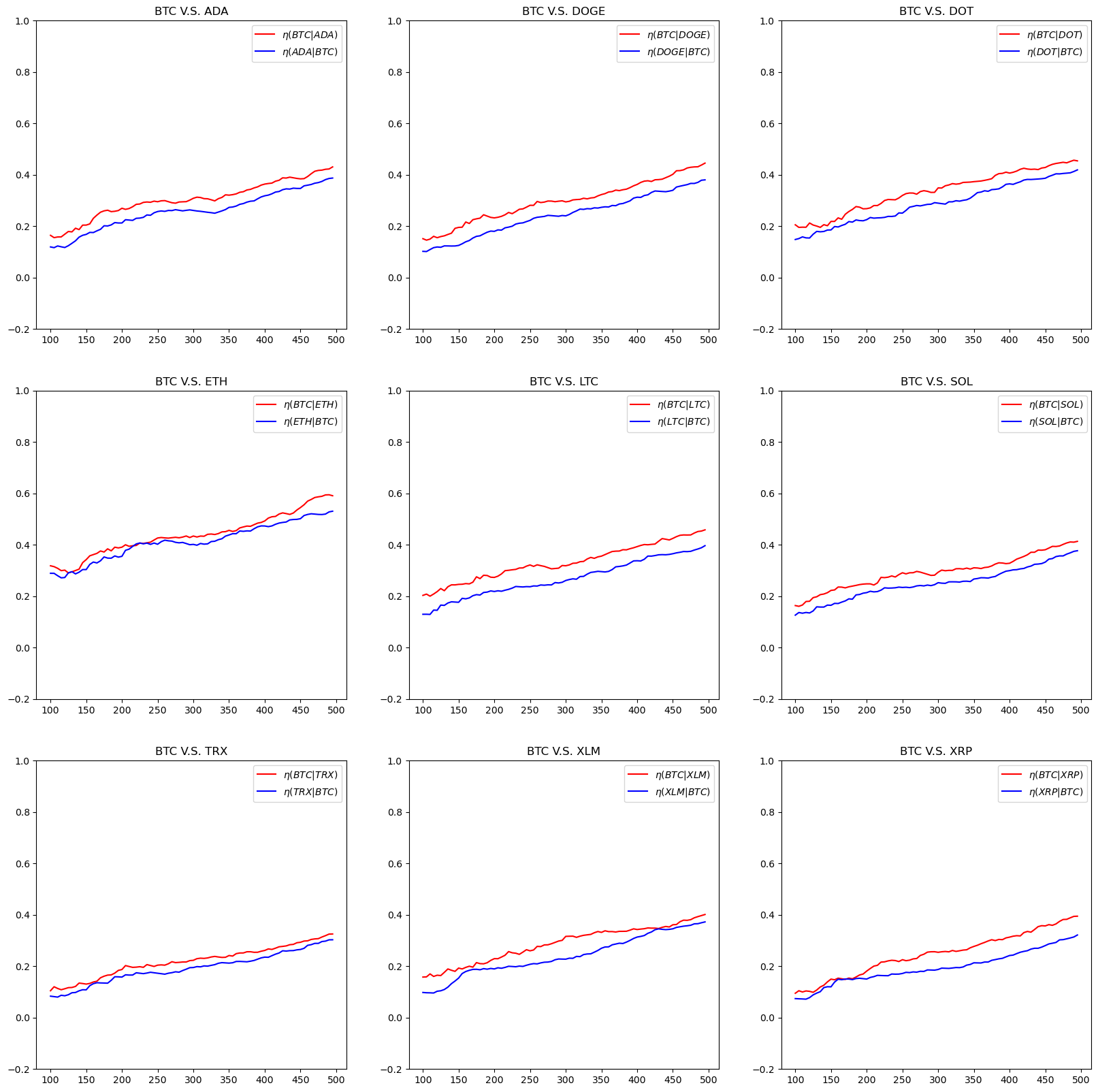} 
     \caption{Cryptocurrency data from \Cref{subsec:datacrypto}: Plot of $\eta_{k,n}(X|Y^{C})$ (red curve) and $\eta(Y^{C}|X)$ (blue curve) for the \textbf{positive} log-returns with different choices of $k\in \mathcal{K}$.
        }\label{fig:eta_compare_pos}
\end{figure}

It has been well-documented that the prices of many highly traded cryptocurrencies are often subject to extreme fluctuations \cite{baur:dimpfl:2021, gong:huser:2022, klein:thu:walther:2018}. The goal of our data analysis is to investigate the asymmetry in tail dependence of cryptocurrency price fluctuations.  More precisely, our hypothesis is that  extreme price increases (or decreases) of certain cryptocurrencies (for example, Bitcoin or Ethereum) may exert a stronger influence on extreme price increases (or decreases) in other cryptocurrencies. We select a representative subset of 10 cryptocurrencies that are among the ones with high market capitalization, namely, Bitcoin (BTC), Litecoin (LTC), XRP, Dogecoin (DOGE), Stellar (XLM), Ethereum (ETH), TRON (TRX), Cardano (ADA), Solana (SOL) and Polkadot (DOT). The data set is obtained using the R package \texttt{crypto2} \citep{R-crypto2} and the complete implementation of our methodology is publicly available at \url{https://github.com/LiuXiangyu99/ETA}. 

\begin{figure}[t]
    \centering
    \includegraphics[width=0.85\textwidth]{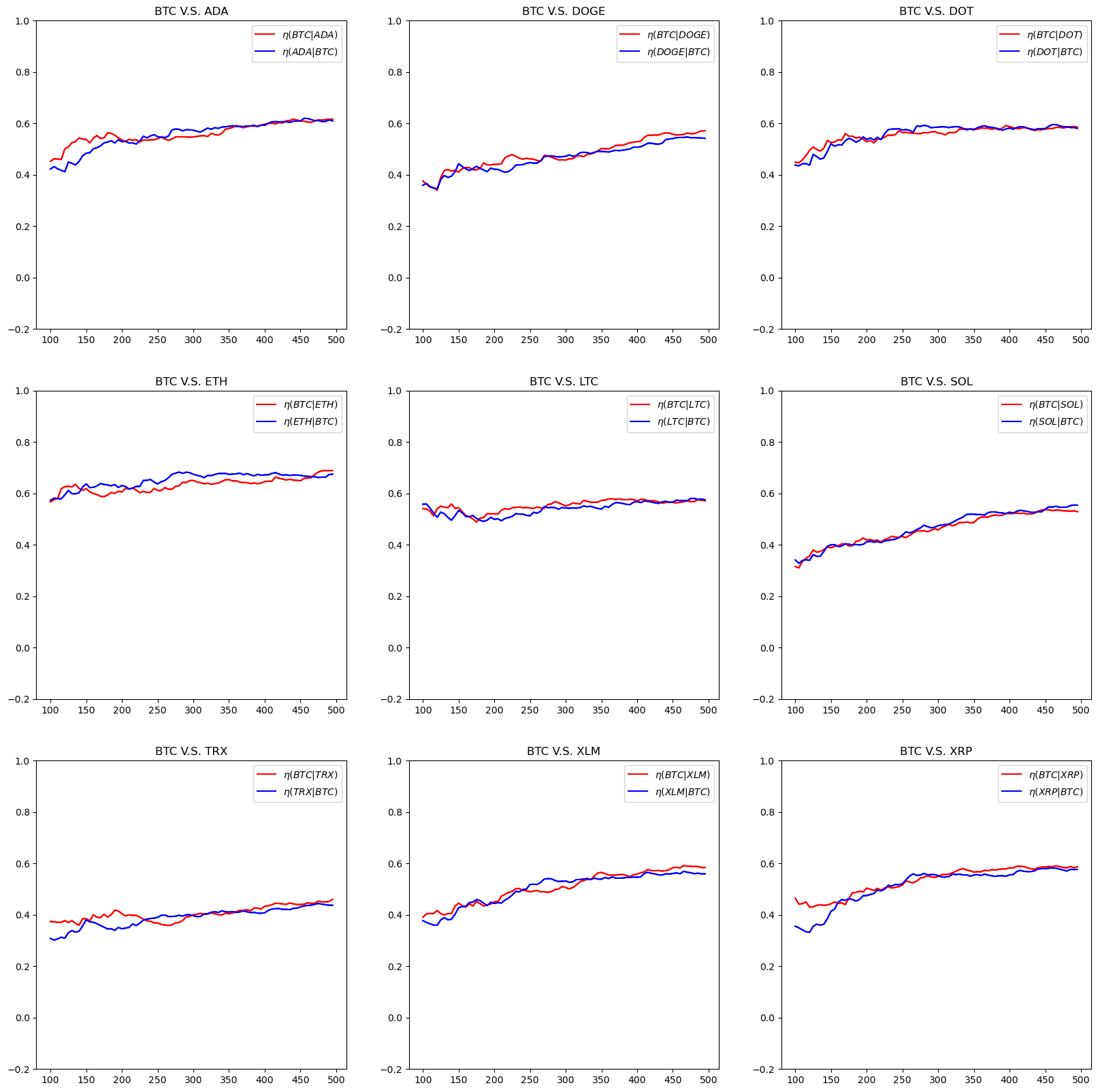} 
        \caption{Cryptocurrency data from \Cref{subsec:datacrypto}: Plot of $\eta_{k,n}(X^{-}|Y^{-C})$ (red curve) and $\eta_{k,n}(Y^{-C}, X^{-})$ (blue curve) for the \textbf{negative} log-returns with different choices of $k\in \mathcal{K}$.   }\label{fig:eta_compare_neg}
\end{figure}
\begin{figure}[t]
    \centering
    \includegraphics[width=0.85\textwidth]{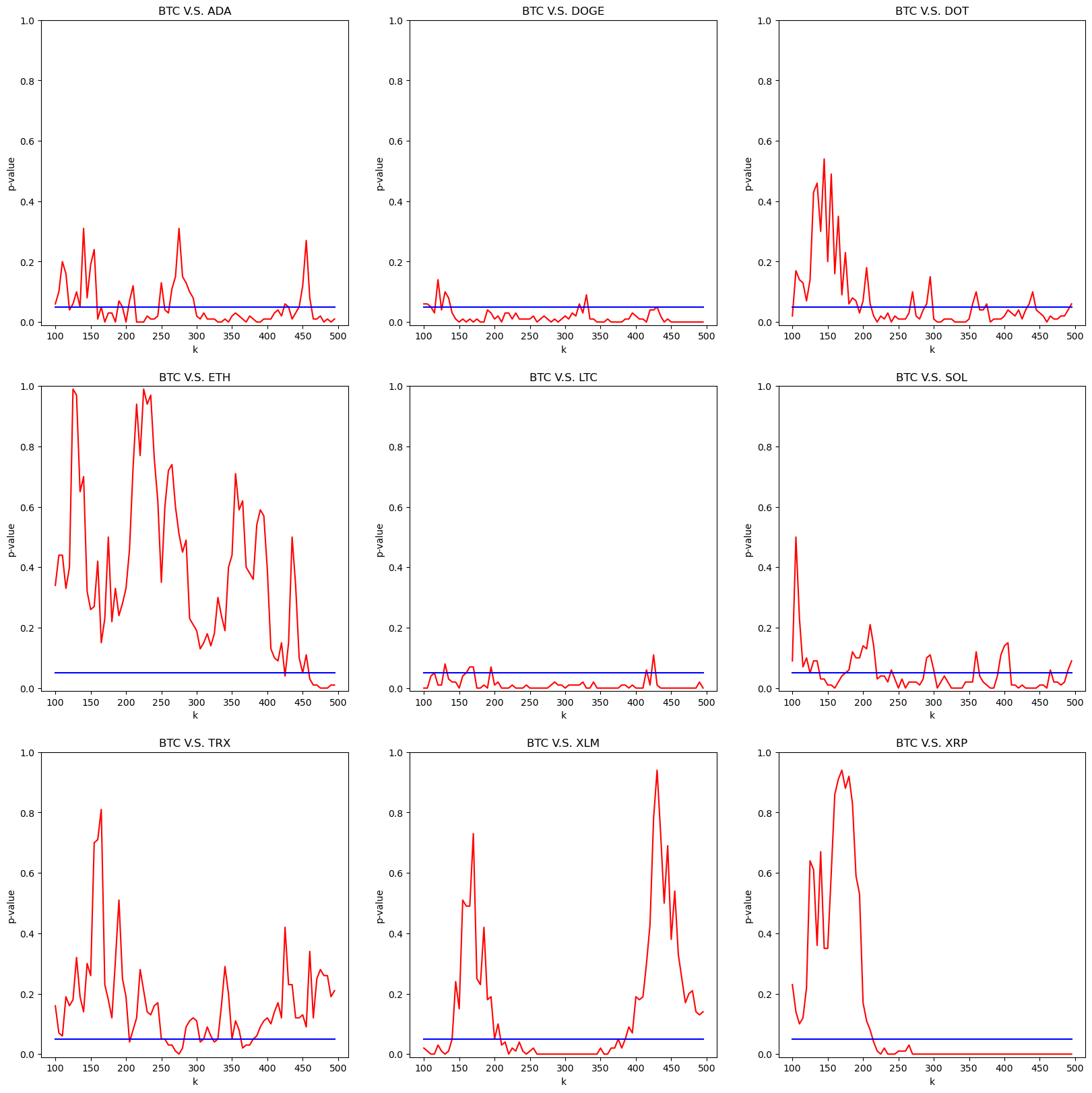} 
    \caption{Cryptocurrency data example: the p-values for testing $\mathcal{H}_0:\Delta(X,Y^C)=0$ versus $\mathcal{H}_1:\Delta(X,Y^C)\neq 0$ in the right tail (extreme positive returns). The red curve plots the p-values as a function of the intermediate sequence $k\in \mathcal{K}$, and the blue horizontal line indicates the reference significance level $\alpha=0.05$.}
    \label{fig:close_pos_boot_pvalue}
\end{figure}

\begin{figure}[t]
    \centering
    \includegraphics[width=0.85\textwidth]{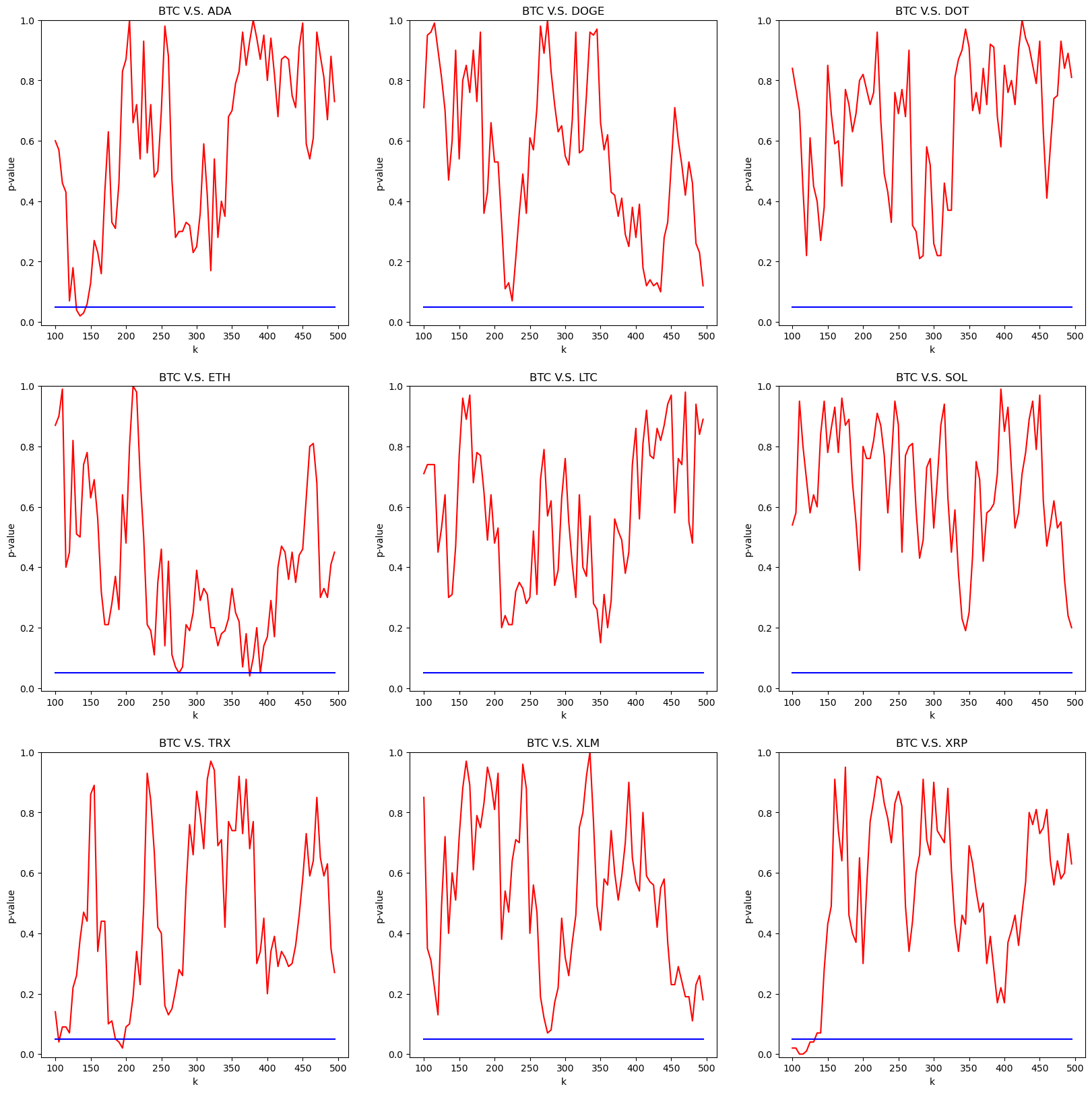} 
    \caption{Cryptocurrency data example: the p-values for testing $\mathcal{H}_0:\Delta(X^{-},Y^{-C})=0$ versus $\mathcal{H}_1:\Delta(X^{-},Y^{-C})\neq 0$ in the left tail (extreme negative returns). The red curve plots the p-values as a function of the intermediate sequence $k\in \mathcal{K}$, and the blue horizontal line indicates the reference significance level $\alpha=0.05$.}
    \label{fig:close_neg_boot_pvalue}
\end{figure}
We have used five years of data, from 1st January, 2021 to 31st December, 2025, on log-returns of closing prices for each of the ten cryptocurrencies mentioned above which represents 1826 days resulting in $n=1825$ log-returns for each cryptocurrency. Financial time series are not necessarily i.i.d., although the test that we created is meant for i.i.d. data. Computing autocorrelations across different lags for each time series, we observe an absence of significant autocorrelation across the lags suggesting that the log-return series indeed resemble an i.i.d. behavior and such an assumption may be reasonable. We have included the autocorrelation plot for each of the time  series in Figure \ref{fig:acf_all}  (\Cref{app:datacrypto}). 

The right tail of returns indicating positive or increasing returns and the left tail of returns indicating negative or decreasing returns are investigated separately. There are $\binom{10}{2}=45$ possible unique pairwise comparisons that can be performed in this case for both positive and negative returns. We hypothesize that Bitcoin (BTC) exerts a higher influence on the other cryptocurrency values for both positive returns and negative returns. Hence we consider 9 pairs of data $(X_1, Y_1^\text{C}), \ldots, (X_n,Y_n^\text{C})$ where $X_i$ indicates (positive) log-returns from Bitcoin on the $i$-th day and $Y_i^\text{C}$ denotes the corresponding log-return of cryptocurrency $\text{C}$ on the same day, with $i=1,\ldots, n=1825$ and $C$ indicating one of the other nine cryptocurrencies, i.e., $\text{C}\in \{\text{ADA, DOGE, DOT, ETH, LTC, SOL, TRX, XLM, XRP}\}$.  We perform the same analysis for negative log-returns of Bitcoin against the log-negative returns of the other nine cryptocurrencies; here the data set for negative log-returns is denoted $(X_1^{-}, Y_1^{-\text{C})}:=(-X_1, -Y_1^\text{C}), \ldots, (X_n^{-}, Y_n^{-\text{C})}:=(-X_n,-Y_n^\text{C})$ for the different values of $C$.

First we compute $\eta(X|Y^C)$ and $\eta(Y^C|X)$ to check if there is tail association between the positive log-returns. We fix the intermediate sequence $$k\in \mathcal{K}=\{100,110, \ldots, 500\}.$$  We then plot the values of $\eta(X|Y^C)$ (in red) and $\eta(Y^C|X)$ (in blue) varying $k\in \mathcal{K}$ for the nine choices of $C$ in \Cref{fig:eta_compare_pos}. The tail dependence between the positive log-returns is evident, and the red curve seems to be above the blue curve most of the time, indicating an asymmetric relationship with Bitcoin having a higher influence on the other cryptocurrency.  We do the same computation for $\eta(X^-|Y^{-C})$ and $\eta(Y^{-C}|X^-)$ to check tail association for negative log-returns and the plot of $\eta(X^-|Y^{-C})$ (in red) and $\eta(Y^{-C}|X^-)$ (in blue) for $k\in \mathcal{K}$ for the nine choices of $C$ is given in \Cref{fig:eta_compare_neg}. In this case the tail dependence is still evident (non-zero value), but the blue and red curves seem to be quite close. Please see Figures \ref{fig:etaXY_boot_pos}-\ref{fig:etaYX_boot_neg} in \Cref{app:datacrypto} where we have 95\% confidence bounds around the estimates of $\eta$ each of which are far away from 0, indicating tail dependence (at least at 5\% level).  Thus, we encounter a scenario where we have tail dependence between variables but would like to test whether the influence is asymmetric or not; we carry this out next.

We perform nine tests of hypotheses for asymmetry in positive log-returns: $$\mathcal{H}_0: \Delta(X,Y^\text{C})=0\quad \text{vs.} \quad \mathcal{H}_1: \Delta(X,Y^\text{C})\neq 0,$$ for the nine different choices of $C$. We  employ the multiplier bootstrap method from \Cref{sec:testasy} to compute p-values for the tests by creating $B=100$ multiplier bootstrap samples each time. The p-values are plotted in  \Cref{fig:close_pos_boot_pvalue} for the nine different hypotheses tests, the blue line in each plot indicates the threshold value of $0.05$ and the red lines are the p-values across all values of $k\in \mathcal{K}$. These p-values indicate that at 5\% level, we would reject the hypothesis of $\mathcal{H}_0: \Delta(X,Y^\text{C})=0$ for $\text{C}\in \{\text{DOGE, DOT, LTC, XRP}\}$ since in these cases the p-values are below $\alpha=0.05$ for around 75\%-80\% or more values of $k\in\mathcal{K}$. When $\text{C} \in $ $\{\text{ADA, XLM, SOL}\}$, the p-values are intermittently above and below the line $\alpha=0.05$ but there seems to be enough evidence to reject $\mathcal{H}_0$. For the other two cryptocurrencies Ethereum and TRON, we cannot find enough evidence to reject $\mathcal{H}_0$. Moreover, from the estimated value of $\Delta(X,Y^{\text{C}})$, we observe that it is positive, indicating that when $\mathcal{H}_0$ is rejected, a large increase in the closing price of Bitcoin has a greater influence on the large increase in the closing price of the other cryptocurrency supporting the traditional belief about Bitcoin. We have provided the estimated value of $\Delta(X,Y^{\text{C}})$ and a 95\% confidence band around it for values of $k\in \mathcal{K}$ is given in \Cref{fig:close_pos_boot} (in \Cref{app:datacrypto}).

Now we consider testing asymmetry for negative log-returns. The p-values for testing $\mathcal{H}_0: \Delta(X^-,Y^{-\text{C}})=0\quad \text{vs.} \quad \mathcal{H}_1: \Delta(X^-,Y^{-\text{C}})\neq 0$ for  the negative returns are plotted in \Cref{fig:close_neg_boot_pvalue} for $k\in \mathcal{K}$. Here it seems the p-value tends to be more than 0.05 in most cases and there appears to be not enough evidence to indicate that there is any directional influence in the tails between Bitcoin and any of the other nine cryptocurrencies. The estimates of $\Delta(X^{-},Y^{-\text{C}})$ with 95\% confidence band
around it for values of $k\in \mathcal{K}$ is given in \Cref{fig:close_neg_boot} (in \Cref{app:datacrypto}), which also indicates that there is not enough evidence to reject that the value of $\Delta$ is 0.

In summary, we may conclude that for highly traded cryptocurrency there is evidence that a high fluctuation in the increasing direction for the price of Bitcoin seems to have a significant (positive) influence on high fluctuation in the increasing direction for the price of the other cryptocurrencies (particularly DOGE, DOT, LTC, XRP). On the other hand the same is not observed for extreme fluctuation in the negative direction.

\section{Conclusion} \label{sec:conclusion}

The key contribution of this paper is the introduction of the Extreme tail association measure $\eta(X|Y)$, which captures directional tail association for a bivariate random vector $(X,Y)$. The sample estimator is shown to be consistent and we have established $\Delta(X,Y)=\eta(X|Y)-\eta(Y|X)$  as a measure of asymmetry in tail dependence. The effectiveness of $\eta$ and $\Delta$ in assessing tail association and tail asymmetry in extreme fluctuations of cryptocurrency prices is established in our data analysis.

We believe that this work creates a pathway for multiple directions of future research to innovate and improve such measures of tail association. Topics that may be of interest to explore further are listed below, some of them being under active investigation by the authors.
\begin{enumerate}
    \item The measures $\eta(X|Y)$ and $\Delta(X,Y)$ are both functionals of the tail copula $\Lambda$ that indicate that there may be other functionals of the tail copula that can capture such directional association and asymmetry. The benefit of such statistics being functions of the tail copula makes them amenable to consistent estimation.
    \item Although we have $0\le \eta\le 1$ and $-1\le \Delta\le 1$ and the boundary values can be shown to be theoretically attainable, in most examples we observe the value of $\Delta_{k,n}$ (and $\Delta$) to be relatively small. This leads to some difficulty in hypothesis testing. We may surmise that there may be other functionals of the tail copula that can avoid this issue.
    \item We have seen that $\eta_{k,n}$ is consistent for $\eta$; but the data analysis and graphs  indicate a finite-sample bias that may be due to the form of the estimator $\eta_{k,n}$ and its dependence on the intermediate sequence $k$; see Figures \ref{fig:eta_compare_pos}, \ref{fig:eta_compare_neg}. It is worth investigating whether such finite-sample bias can be corrected by appropriately modifying the estimator.
    \item An asymmetry in tail association can be observed in many real datasets, yet, there are very few popularly used parametric models that capture such dependency. It will be of interest to investigate the generating mechanism for such tail asymmetry via parametric models.
    \item Finally, it is not clear what can be said about directional tail association if we observe asymptotic tail independence in data; one such classic example is the Gaussian copula which is often used in practice. An idea will be to look at hidden association measures as the coefficient of tail dependence (\citet{ledford:tawn:1996}), or, hidden regular variation (\citet{resnick:2002}).
\end{enumerate}

\bibliographystyle{imsart-nameyear}
\bibliography{bibfilenew}

\newpage

\appendix

\section{Proof of results in Section \ref{sec:etatailcop}}\label{app:proofsec2}
\begin{proof}[Proof of \Cref{prop:eta_is_intutc}]\label{pf:prop:eta_is_intutc}
\begin{enumerate}[(a)]
\item For the copula $C$ of $(X,Y)\sim F$, using \eqref{eqdef:eta_pop}, \eqref{eq:DenDeltat}, and \eqref{eq:cop:survcop} we have for $0<t\le 1$,
\begin{align}
\eta_t(X|Y) & = \frac{3}{t^2(1-t)^3} \int_t^1 (C(u,t)-ut)^2 \; \mathrm du \nonumber \\
          & =  \frac{3}{t^2(1-t)^3} \int_0^{1-t} (\ov{C}(u',1-t)-u'(1-t))^2 \; \mathrm du'  \label{eq:etatrep2}
\end{align}
where we used $u'=1-u$. Let $s= \frac{1}{1-t}$ and  $v=us$, then  we have
\begin{align}
\eta_{1/(1-s)}(X|Y)   
             & = 3 \int_0^1 \frac{(s\overline{C}(v/s,1/s)-v/s)^2}{(1-1/s)^2} \; \mathrm dv =: 3 \int_0^1 A_s(v) \; \mathrm dv \quad (\text{say}). \label{def:As}
\end{align}
Note that for all $v\in[0,1]$,
\[\lim_{s\to\infty} A_s (v) = \lim_{s\to\infty}  \left(\frac{1}{1-1/s}\right)^2 \left(s\ov{C}(v/s,1/s)- \frac{v}{s}\right)^2= [\Lambda(v,1)]^2.\]
Moreover, for any $s\ge 2$, since $\ov{C}(x,y) \le \min(x,y)$, we have for any $0\le v\le 1$,
\begin{align*}
  |A_s(v)| & \le  \left(\frac{1}{1-1/s}\right)^2 \left(\left|s\ov{C}(v/s,1/s)\right| +\left|\frac{v}{s}\right|\right)^2 \le 4 \cdot \left(1+ \frac{1}{2}\right)^2 =9,
  \end{align*}
Hence, we have by bounded convergence theorem,
  \begin{align*}
  \eta(X|Y)   =  \lim_{s\to\infty} \eta_{1-1/s}(X|Y)=  \lim_{s\to\infty} 3 \int_0^1 A_s(v) \mathrm dv = 3\int_0^1 [\Lambda(v,1)]^2 \mathrm d v.
  \end{align*}
  \item Clearly $\eta(X|Y)\ge0$ by the representation from part (a). Now note that for $u\in [0,1]$,
    \begin{align*}
    \Lambda[u,1] = \lim_{s\to\infty} s\ov{C} (u/s, 1/s) \le s\cdot\frac us =u.
    \end{align*}
    Hence
      \begin{align*}
    \eta(X|Y) = 3 \int_0^1 (\Lambda[u,1])^2 \mathrm du \le 3 \int_0^1 u^2 \mathrm du \le 1.
    \end{align*}
    
  \end{enumerate}

\end{proof}

Before proving \Cref{prop:uptc_symtoasym}, we state the following lemmas that guaranty the continuous convergence condition required for its proof.

\begin{lemma}[Lipschitz continuity \cite{Nelsen}] \label{lem:lip_cont}
    Suppose $C:[0, 1]^2 \to [0, 1]$ is a copula. Then, for $(u,v), (u',v')\in [0,1]^2$ we have
    \begin{align*}
        \left|C(u, v) - C(u^{'}, v^{'})\right| \leq |u - u^{'}| + |v - v^{'}|.
    \end{align*}
\end{lemma}

\begin{lemma} \label{lem:cont_converg}
    Suppose $C:[0, 1]^2 \to [0, 1]$ is a copula. Then for any point $(x, y) \in [0, 1]$ and any perturbations $r_i(t) = o(t^{-1})$ where $i = 1, 2$,
    \begin{align*}
    t\left[C\left(1-\frac{x}{t}+r_1(t), 1-\frac{y}{t}+r_2(t)\right) - C\left(1-\frac{x}{t}, 1-\frac{y}{t}\right)\right] \to 0
    \end{align*}
 as $t\to\infty$ and this convergence is uniform on $[0, 1]^2$.
\end{lemma}

\begin{proof}
    From \Cref{lem:lip_cont}, we know $C$ is Lipschitz continuous and hence
    \begin{align*}
        \left|C\left(1 - \frac{x}{t} + r_1(t),\;
                1 - \frac{y}{t} + r_2(t)
            \right) - C\left(1 - \frac{x}{t}, 1 - \frac{y}{t}\right)  \right| \leq  \left|r_1(t)\right| + \left|r_2(t)\right|.
    \end{align*} Now by multiplying $t$ on both sides we obtain the desired result.
\end{proof}

\begin{proof}[Proof of \Cref{prop:uptc_symtoasym}]\label{pf:prop:uptc_symtoasym}
For a symmetric copula $C^*$, we have for $(x,y)\in [0,1]^2$,
\[
C_{\alpha,\beta}(x,y) := x^{1-\alpha} y^{1-\beta} C^*(x^{\alpha}, y^{\beta}), 
\]
Hence for $(x,y)\in [0,1]^2$, the corresponding  survival copula $\overline{C}_{\alpha,\beta}$ is given by
\begin{align}
\overline{C}_{\alpha,\beta}(x,y)
&= x + y - 1 + C_{\alpha,\beta}(1 - x, 1 - y) \nonumber\\
&= x + y - 1 
  + (1 - x)^{1-\alpha} (1 - y)^{1-\beta}
    C^*\left( (1 - x)^{\alpha}, (1 - y)^{\beta} \right). \label{eq:cbartocstar}
\end{align}
For large enough $t$, with $(u,v)\in [0,1]^2$, we have
\begin{align*}
& \left( 1 - \frac{u}{t} \right)^{1-\alpha}
\left( 1 - \frac{v}{t} \right)^{1-\beta}
C^*\left(
  \left( 1 - \frac{u}{t} \right)^{\alpha},
  \left( 1 - \frac{v}{t} \right)^{\beta}
\right)\\
 & =  \left( 1 - \frac{(1-\alpha)u+(1-\beta)v}{t} + o(t^{-1}) \right)
  C^*\left(
    1 - \frac{\alpha u}{t} + o(t^{-1}),
    1 - \frac{\beta v}{t} + o(t^{-1})
  \right)\\
& = C^*\left(
    1 - \frac{\alpha u}{t} + o(t^{-1}),
    1 - \frac{\beta v}{t} + o(t^{-1})
  \right)
  - \left( \frac{(1-\alpha)u + (1-\beta)v}{t} \right) (1 + o(1)),
\end{align*}
and thus
\begin{align}
t\,\overline{C}_{\alpha,\beta}\left(\frac ut,\frac vt\right)
& = t\left( \frac ut + \frac vt -1 \right) + t\,C^*\left(
    1 - \frac{\alpha u}{t} + o(t^{-1}),
    1 - \frac{\beta v}{t} + o(t^{-1})
  \right) \nonumber \\
  & \quad\quad - t\,\left( \frac{(1-\alpha)u + (1-\beta)v}{t} \right) (1 + o(1)) \nonumber\\
  & = \alpha u +\beta v -t + t\,C^*\left(
    1 - \frac{\alpha u}{t} + o(t^{-1}),
    1 - \frac{\beta v}{t} + o(t^{-1})
  \right) + o(1) \nonumber\\
  & = t\left( \frac{\alpha u}t + \frac{\beta v}t -1\right) + t\,C^*\left(
    1 - \frac{\alpha u}{t},
    1 - \frac{\beta v}{t}
  \right) + o(1) \nonumber\\
  & = t\,\overline{C}^*\left(\frac {\alpha u}t,\frac {\beta v}t\right) + o(1), \label{eq:cbartocstarbar}
\end{align}
where the final equality above results from \eqref{eq:cbartocstar} and the second to last equality is justified by \Cref{lem:cont_converg}.
Now by taking the limit as $t\to\infty$ for both LHS and RHS in \eqref{eq:cbartocstarbar}, we obtain
\begin{align*}
 \Lambda_{C_{\alpha,\beta}}(u,v) = \Lambda_{C^*}(u,v), \quad (u,v)\in[0,1]^2.
\end{align*}

\end{proof}

\section{Proof of results in Section \ref{sec:asybeh}}\label{app:proofsec3}
The results in \Cref{sec:asybeh} relate to convergence of the measures $\eta_{k,n}$ and $\Delta_{k,n}$. First we provide some basic definitions and concepts useful for proving these convergence results, especially the weak convergence of empirical tail copulas. We also recall the relevant functional delta method we use to obtain said convergence of the statistics mentioned in these papers. Some of the concepts and are adapted from \citet{bucher:dette:2013,vandervaart:wellner:1996,schmidt:stadtmuller:2006}.
\subsection{Preliminaries}\label{subsec:prelim}
Denote by $\overline{\R}_+^2:=[0,\infty]^2\setminus\{(0,0)\}$ and let $\mathcal{B}_{\infty}(\overline{\R}^2_+)$
 be the (metric) space of functions $f: \overline{\R}_+^2 \to \R$ which are locally uniformly bounded on  compact subsets of $\overline{\R}_+^2$ (these are closed subsets of $[0,\infty]^2$ which are bounded away from $\{(0,0)\}$) together with the complete metric:
    \begin{align*}
        d(f_1, f_2) := \sum_{i = 1}^{\infty} 2^{-i} \min(\lVert f_1 - f_2 \rVert_{T_i}, 1)
    \end{align*}  where we have  $T_i=[0,i]^2\bigcup [0,i]\times\{\infty\}\bigcap\{\infty\}\times[0,i]$, and 
    $||f||_{T_i}=\sup_{\bx\in T_i} |f(\bx)|$ denotes the sup-norm on $T_i$. We also denote by $\mathcal{C}(\overline{\R}_+^2)$, the set of continuous functions in $\overline{\R}_+^2$ and naturally $\mathcal{C}(\overline{\R}_+^2)\subset\mathcal{B}(\overline{\R}_+^2)$.

\begin{definition}[Hadamard Differentiability \cite{vandervaart:wellner:1996}]
   Suppose $\mathbb{D}$ and $\mathbb{E}$ are two metrizable topological vector spaces and consider the map $\phi: \mathbb{D}_{\phi} \subseteq \mathbb{D}  \to \mathbb{E}$. We call $\phi$ to be Hadamard-differentiable at $\theta \in \mathbb{D}_{\phi}$ if there is a continuous linear map $\phi^{'}_{\theta}: \mathbb{D} \to \mathbb{E}$ such that:

\begin{align*}
    \frac{\phi(\theta + t_nh_n) - \phi(\theta)}{t_n} \to \phi^{'}_{\theta}(h), \qquad n \to \infty
\end{align*}
for all converging real sequences $t_n \to 0$ and all $\mathbb{D}$-valued sequences $h_n \to h$ such that $\theta + t_nh_n \in \mathbb{D}_{\phi}$. The definition may be refined to Hadamard-differentiable tangentially to a set $\mathbb{D}_0 \subseteq \mathbb{D}$ by requiring that every $h_n \to h$ in the definition has $h \in \mathbb{D}_0$; the domain of $\phi^{'}_{\theta}$ need then be on $\mathbb{D}_0$ only.

\end{definition}

Under such regularity of $\phi$, we can apply the generalized Delta-method:

\begin{theorem}[Functional Delta Method \cite{vandervaart:wellner:1996}]\label{thm:deltamethod}
    Let $\mathbb{D}$ and $\mathbb{E}$ be metrizable topological vector spaces. Let $\phi: \mathbb{D}_{\phi} \subseteq \mathbb{D}  \to \mathbb{E}$ be Hadamard-differentiable at $\theta$ tangentially to $\mathbb{D}_0\subset \mathbb{D}$. Let $X_n: \Omega_n \to \mathbb{D}_{\phi}$ be maps with $r_n(X_n - \theta) \weak X$ for some sequence of constant $r_n \to \infty$, where $X$ is separable and takes its values in $\mathbb{D}_0$. Then $r_n(\phi(X_n) - \phi(\theta)) \weak \phi^{'}_{\theta}(X)$. Here $\weak$ represents weak convergence in Hoffman-J\o rgensen sense (\citep{vandervaart:wellner:1996}).
    
\end{theorem}

\subsection{Proofs}
First we prove Proposition \ref{thm:consistency}. The proof requires results adapted from \cite{qi:1997} and \cite[Theorem 6]{schmidt:stadtmuller:2006} showing strong (and weak) consistency of empirical estimate for the tail copula (equivalently the stable tail dependence function), a representation of $\eta$ from the the proof of \Cref{prop:eta_is_intutc} (provided later) and some auxiliary results. The following lemma  is a consequence of \cite[Theorem 1.2]{qi:1997}.
\begin{lemma}\label{lem:lambdanconvergence}
      Let $(X_1,Y_1), \ldots, (X_n,Y_n)$ be i.i.d from a bivariate distribution $F$ with continuous marginals $F_X, F_Y$. Suppose $\Lambda\neq 0$ exists (as defined in \eqref{eq:tailcopup}) and is estimated by $\widehat{\Lambda}_{k,n}$ (as defined in \eqref{eq:lambdahat}). Then as $n\to \infty, k \to \infty$ and $n/k \to \infty$, we have for any $T>0$,
 \begin{align}\label{eq:lambdahatinP}
     \sup_{(x,y)\in [0,T]^2} |\widehat\Lambda_{k,n}(x,y)-\Lambda(x,y)| \to 0 \quad \text{in probability.}
 \end{align}
Additionally, if $k/\log\log n \to \infty$ as $n\to \infty$, then
  \begin{align}\label{eq:lambdahatas}
    \sup_{(x,y)\in [0,T]^2} |\widehat\Lambda_{k,n}(x,y)-\Lambda(x,y)| \to 0 \quad \text{almost surely.}
 \end{align}
   
\end{lemma}
\begin{proof}
For $(x,y)\in \R_+^2$, the \emph{stable tail dependence function} $\ell$ \cite{huang:1992,dehaan:ferreira:2006} is defined as
\[\ell(x,y) = \lim_{t\to\infty} t(1-F(\ov{F}_X^{\leftarrow}(y/t),\ov{F}_Y^{\leftarrow}(y/t)))\]
if it exists. Clearly $\ell$ as defined exists if $\Lambda$ ( cf.\eqref{eq:tailcopup}) exists and for $(x,y)\in \R_+^2$,
\begin{align}\label{eq:elllambda}
    \Lambda(x,y) = \lim_{t\to\infty} t\ov{F}(\ov{F}_X^{\leftarrow}(y/t),\ov{F}_X^{\leftarrow}(y/t))= x+y - \ell(x,y).
\end{align}
Now following \cite{qi:1997} and adapting to our notation define
\begin{align*}
  \ell_n(x,y) := \frac1k \sum_{i=1}^n  \left(1-F_n(\ov{F}_{n,X}^{\leftarrow}\left(kx/n\right),\ov{F}_{n,Y}^{\leftarrow}\left(ky/n\right)))\right)
\end{align*}
Expanding we observe that
\begin{align}\label{eq:ellnlambdan}
    \ell_n(x,y) = \frac 1k \left\lceil{kx}\right\rceil +\frac 1k \left\lceil{ky}\right\rceil - \widehat{\Lambda}_{k,n}(x,y).
\end{align}
where for any $z\ge 0$, $\ceil{z}$ is the smallest integer larger than $z$.
Therefore, from \eqref{eq:elllambda} and \eqref{eq:ellnlambdan}, we have for any $x,y \ge 0$,
\begin{align*}
    \big|\widehat{\Lambda}_{k,n}(x,y)-\Lambda(x,y)\big| & \le |\ell_n(x,y)-\ell(x,y)| + \Big| \frac 1k \left\lceil{kx}\right\rceil +\frac 1k \left\lceil{ky}\right\rceil -(x+y)\Big|\\
    & \le |\ell_n(x,y)-\ell(x,y)| + \frac 2k.
\end{align*}
  Now \eqref{eq:lambdahatinP} and \eqref{eq:lambdahatas} follows from \cite[Theorem 1.2]{qi:1997} and since $2/k\to 0$ as $k\to \infty$.
\end{proof}

\begin{proof}[Proof of \Cref{thm:consistency}]
First we show that $\eta_{k,n}$ is the natural estimator of $\eta(X|Y)$ replacing $\Lambda$ by its empirical counterpart (cf. \eqref{eq:lambdahat}).
 Expanding the terms for $\widehat{\Lambda}_{k,n}(u,1)$ for $0\le u \le 1$  we have
     \begin{align*}
        \widehat{\Lambda}_{k,n}(u,1) & = \frac nk \cdot \frac 1n \sum_{i=1}^n \bI(X_i > \ov{F}_{n,X}^{\leftarrow}(ku/n), Y_i > \ov{F}_{n,Y}^{\leftarrow}(k/n))\\
                           & = \frac 1k \sum_{i=1}^{k-1} \bI(\ov{F}_{n,X}(X_{[i]})<ku/n) = \frac 1k \sum_{i=1}^{k-1} \bI(\rho_i<ku+1),
     \end{align*}
     where the final equality follows from the definition of $\rho_i$ (see \eqref{def:revrank}) and the second to last equality follows from the definition of concomitant $X_{[i]}$ of $Y_{(i)}$. Hence
    \begin{align*}
        3\int_0^1 \left[ \widehat{\Lambda}_{k,n}(u,1)\right]^2 \mathrm du & = \frac3{k^2}\int_0^1 \sum_{i=1}^{k-1}\sum_{j=1}^{k-1} \bI(\rho_i<ku+1, \rho_j<ku+1) \, \mathrm d u\\
                           & = \frac 3{k^3} \sum_{i=1}^{k-1}\sum_{j=1}^{k-1} (k+1 - \max(\rho_i,\rho_j))_+ = \eta_{k,n}(X|Y).
     \end{align*} 
Now given $\epsilon>0, \delta>0$ small, for large enough $k, n$ and $n/k$,  using \eqref{eq:lambdahatinP} (\Cref{lem:lambdanconvergence}),  we have
 \begin{align}\label{eq:lambdaepsilondelta}
 \P\left(\sup_{0\le  u\le 1} \left|[\widehat{\Lambda}_{k,n}(u,1)]^2-[\Lambda(u,1)]^2\right|> \epsilon/3 \right) <\delta.
 \end{align}
Therefore,
 \begin{align*}
     \P\left(|\eta_{k,n}(X|Y) - \eta(X|Y)|> \epsilon \right) & = \P\left(\left|3\int\limits_0^1 \widehat{\Lambda}_{k,n}(u,1)]^2 \,\mathrm d u - 3\int\limits_0^1 [\Lambda(u,1)]^2 \,\mathrm d u \right|> \epsilon \right) \\
       & \le  \P\left(3\int\limits_0^1 \left|\widehat{\Lambda}_{k,n}(u,1)]^2 - [\Lambda(u,1)]^2 \right| \, \mathrm d u > \epsilon \right)\\
       & \le \P\left(3\int\limits_0^1  \mathrm d u \times \sup_{0\le  u\le 1} \left|\widehat{\Lambda}_{k,n}(u,1)]^2-[\Lambda(u,1)]^2\right| > \epsilon \right)\\
       & <\delta.
 \end{align*}
  Hence, as $n\to \infty, k \to \infty$ and $n/k \to \infty$, we have
\begin{align}\label{eq:eta*consistent}
    {\eta}_{k,n}(X|Y) \to \eta(X|Y) \quad \text{in probability}. 
\end{align}
Almost sure convergence of ${\eta}_{k,n}(X|Y)$  can be proved similarly with the additional assumption of $k/\log\log n \to\infty$ and using \eqref{eq:lambdahatas}, the proof is skipped for brevity.
\end{proof}

To prove the asymptotic normality of the estimators $\eta_{k, n}$ and $\Delta_{k,n}$, we will utilize Proposition \ref{thm:bucherdette} and then apply the functional delta method (Theorem \ref{thm:deltamethod}) to the empirical tail copula $\widehat{\Lambda}_{k,n}$.
 First we  introduce the functionals $\phi,\psi: \mathcal{B}_{\infty}(\overline{\mathbb{R}}^2_+) \to \mathbb{R}$ as follows
    \begin{align}
        &\phi(f) := 3\int_{0}^1 [f(u, 1)]^2 \mathrm{d}u, \label{def:phi}\\
         &\psi(f) := 3\int_{0}^1 [f(u, 1)]^2 \mathrm{d}u - 3\int_{0}^1 [f(1, v)]^2 \mathrm{d}v. \label{def:psi}
    \end{align}
\begin{proposition}\label{prop:diff}
 The functionals $\phi$ and $\psi$ as defined in \eqref{def:phi} and \eqref{def:psi} are continuous and Hadamard-differentiable at $\Lambda$ tangentially to the continuous functions $\mathcal{C}(\overline{\mathbb{R}}^2_+) \subset \mathcal{B}_{\infty}(\overline{\mathbb{R}}^2_+)$. Moreover their derivatives  take the form:
    \begin{align}
        \phi'_{\Lambda}(h) & = 6\int_0^1 \Lambda(u, 1)h(u, 1) \mathrm du, \label{eq:phi}\\
                \psi'_{\Lambda}(h) & = 6\int_0^1 \Lambda(u, 1)h(u, 1) \mathrm du -6\int_0^1 \Lambda(1,v)h(1,v) \mathrm dv. \label{eq:psi}
    \end{align}    
\end{proposition} 

\begin{proof}
  The continuity of $\phi$, and hence $\psi$ is quite evident since we obtain them by integrating a locally bounded function on compact intervals. 
 In order to show Hadamard differentiability at $\Lambda$, take any real (positive) converging sequence $t_m \to 0$ as $m\to \infty$; a sequence of functions $h_m\in\mathcal{B}_{\infty}(\overline{\mathbb{R}}^2_+), m\ge 1$, and a function $h\in \mathcal{C}(\overline{\mathbb{R}}^2_+)$ such that $h_m\to h$ as $m\to \infty$ with
    \begin{align*}
        \Lambda_m(x, y) = \Lambda(x, y) + t_mh_m(x, y).
    \end{align*}
    Then we have
    \begin{align*}
        \frac{\phi(\Lambda_m) - \phi(\Lambda)}{t_m} & = 3t_m^{-1} \int_0^1 \left([\Lambda_m(u, 1)]^2 - [\Lambda(u,1)]^2\right) \mathrm{d}u \\
        &= 3 \int_0^1 \left(2\Lambda(u,1) + t_mh_m(u, 1)\right) h_m(u, 1)\mathrm{d}u \\
        &= 6 \int_0^1 \Lambda(u,1)h_m(u, 1)\mathrm{d}u + 3 t_m \int_0^1 h_m(u, 1)^2 \mathrm{d}u =: A_m + B_m \quad \text{(say)}.
    \end{align*}
    As $m\to \infty$, since $h_m \to h$ which is continuous and hence bounded on $[0, 1]^2$, this implies that $B_m\to0$ as $m\to \infty$. The Dominated Convergence Theorem (see \cite{billingsley:1968}) guaranties the convergence of $A_m$, and hence we have:
    \begin{align*}
     \phi^{'}_{\Lambda}(h) =  \lim_{m\to\infty}\frac{\phi(\Lambda_m) - \phi(\Lambda)}{t_m} 
        = 6 \int_0^1 \Lambda(u,1)h(u, 1)\mathrm{d}u,
    \end{align*}
 proving \eqref{eq:phi}. 
    The proof of equality (convergence) in \eqref{eq:psi} is similar and is skipped here.

\end{proof}

\begin{proof}[Proof of \Cref{thm:asymnorm}]

    Let $X_k = \hat{\Lambda}_{k,n}$, $r_k = \sqrt{k}$ and $X = \Lambda$. Due to  \Cref{prop:diff} we may apply the Functional Delta-method (Theorem \ref{thm:deltamethod}) on  \Cref{prop:bootstrappingEmpiricalProcesses} with the function $\phi$ as defined in \eqref{eq:phi} and as $, n\to\infty, k\to \infty, n/k\to\infty$ we have Theorem \ref{thm:asymnorm}.
\end{proof}

\begin{proof} [Proof of Theorem \ref{thm:teststatdist}]
 The proof is very similar to that of Theorem \ref{thm:asymnorm}. Here we apply the Functional Delta-method (Theorem \ref{thm:deltamethod}) on  \Cref{prop:bootstrappingEmpiricalProcesses} with the function $\psi$ as defined in \eqref{eq:psi} and the rest remains the same as in Theorem \ref{thm:asymnorm}. 
\end{proof}

\section{Proof for results in Section \ref{sec:testasy}} \label{app:proofsec4}

We use the functional delta method again to prove the convergence of bootstrap estimators following \cite{bucher:dette:2013}. Denote by $\mathbf{S}_n = \{(X_1,Y_1), \ldots, (X_n,Y_n),(\xi_1, \ldots, \xi_n)\}, n\ge 1$. With random variables $H_n(\mathbf{S}_n), H \in \mathbb{D}$ for a  metric space $(\mathbb{D},d)$, we write $H_n \weak_\xi H$ as $n\to\infty$ to denote conditional weak convergence in $(\mathbb{D},d)$  given the data $\{(X_i,Y_i)\; 1\le i \le n\}$ as defined in \cite{Kosorok2008}; see \cite[Section 3.1]{bucher:dette:2013} for more details. We write $\E_{\xi}$ to denote the conditional expectation over $\xi_1, \ldots, \xi_n$ given the data $\{(X_i,Y_i)\;1\le i \le n\}$.

\begin{lemma}\label{lemma:condweaklemma}
    Consider random variables $U_n(\mathbf{S}_n), V_n(\mathbf{S}_n), Z \in \mathbb{D}$, for $n\ge1$. If \newline  $\E_{\xi} \|U_n-V_n\| \stp 0$, and  $V_n \weak_{\xi} Z$, then we have $U_n\weak_{\xi} Z$.
\end{lemma}
\begin{proof}
    Let $f \in \text{BL}_1$, where  \[
\text{BL}_1(\mathbb{D}) := \left\{ f : \mathbb{D} \to \mathbb{R}\Big|\, \|f\|_\infty \le 1,\ \text{and}, |f(x)-f(y)|\le d(x,y) \, \text{ for } \, x,y \in \mathbb{D} \right\}.\] 
Now,
    \begin{align*}
        |\E_{\xi} f(U_n) - \E_{\xi} f(Z)| 
        &= |\E_{\xi} f(V_n+U_n-V_n) -\E_{\xi} f(V_n)+\E_{\xi} f(V_n)- \E_{\xi} f(Z)|\\
        &\leq |\E_{\xi} f(V_n+U_n-V_n) -\E_{\xi} f(V_n)|+|\E_{\xi} f(V_n)- \E_{\xi} f(Z)|\\
        &\leq  \mathbb{E}_{\xi} \|U_n-V_n\| + \sup_{f\in \text{BL}_1}|\mathbb{E}_{\xi} f(V_n) - \mathbb{E}f(Z)| \\
        & \stp 0, \quad\text{as $n \to \infty$}.
    \end{align*}
 The rest follow by definition (see \cite[Section 3]{bucher:dette:2013}).   
\end{proof}

\begin{proof}[Proof of theorem \ref{thm:bootstrappingCLT}]
We will prove the Central Limit Theorem for $\Delta^{\xi, \xi}_{k, n}(X, Y)$,   the proof of CLT for $\eta^{\xi, \xi}_{k, n}$ follows in a similar manner and is skipped here. Assume $\mu = \tau = 1$, this can be achieve by taking $\xi_i, i\ge 1$ to be a sequence of i.i.d. standard exponential random variables. Denote \begin{align*}
    \alpha^{\xi,\xi}_{k, n}(u, v) :=& \sqrt{k} (\Lambda^{\xi, \xi}_{k, n}(u, v) - \widehat{\Lambda}_{k, n}(u, v)),\\
    \mathcal{E}_{k,n} (u, v) :=& \sqrt{k} (\widehat{\Lambda}_{k, n}(u, v) - \Lambda(u, v)).
\end{align*}
Then we have: 
\begin{align*}
    & \sqrt{k} (\Delta^{\xi, \xi}_{k, n}(X, Y) - \Delta_{k, n}(X, Y)) \\
    & =3 \sqrt{k} \int_0^1 \left[\Lambda^{\xi, \xi}_{k, n} (u, 1)^2 - \widehat{\Lambda}_{k, n}(u, 1)^2 + \widehat{\Lambda}_{k, n}(1, u)^2 - \Lambda^{\xi, \xi}_{k, n} (1, u)^2\right] \mathrm{d}u\\
 & = 3 \sqrt{k} \int_0^1 \Big[(\Lambda^{\xi, \xi}_{k, n} (u, 1) + \widehat{\Lambda}_{k, n}(u, 1))(\Lambda^{\xi, \xi}_{k, n} (u, 1) - \widehat{\Lambda}_{k, n}(u, 1))  \\&  \qquad \qquad\qquad \qquad - (\Lambda^{\xi, \xi}_{k, n} (1, u) + \widehat{\Lambda}_{k, n}(1, u))(\Lambda^{\xi, \xi}_{k, n} (1, u) - \widehat{\Lambda}_{k, n}(1, u))\Big] \mathrm{d}u \\
    & = 3  \int_0^1 \Bigg[\left(2 \widehat{\Lambda}_{k, n}(u, 1) + \frac{1}{\sqrt{k}}  \alpha^{\xi,\xi}_{k, n}(u, 1)\right)  \alpha^{\xi,\xi}_{k, n}(u, 1) \\ & \qquad \qquad\qquad \qquad- \left(2 \widehat{\Lambda}_{k, n}(1, u) + \frac{1}{\sqrt{k}}  \alpha^{\xi,\xi}_{k, n}(1, u)\right)  \alpha^{\xi,\xi}_{k, n}(1, u)\Bigg]\mathrm{d}u \\
    & = 6 \int_0^1 \left[\widehat{\Lambda}_{k, n}(u, 1) \alpha^{\xi,\xi}_{k, n}(u, 1) - \widehat{\Lambda}_{k, n}(1, u) \alpha^{\xi,\xi}_{k, n}(1, u)\right] \mathrm{d}u \\
     & \qquad \qquad \qquad \qquad+ \frac{3}{\sqrt{k}} \int_0^1 \left[\alpha^{\xi,\xi}_{k, n}(u, 1)^2 - \alpha^{\xi,\xi}_{k, n}(1, u)^2\right] \mathrm{d}u \\
    & =: A^{\xi, \xi}_{k, n} + B^{\xi, \xi}_{k, n}\quad  \text{(say)}.
\end{align*}
Now define $\phi: \mathcal{B}(\overline{\R}^2_+) \times \mathcal{B}(\overline{\R}^2_+)\to\R$ as
\begin{align*}
    \phi(f, g) = {6} \int_0^1 \left[f(u, 1)g(u, 1) - f(1, u)g(1, u)\right] \mathrm{d}u.
\end{align*} Then we have \begin{align*}
    A^{\xi, \xi}_{k, n} = \phi(\widehat{\Lambda}_{k, n}, \alpha^{\xi,\xi}_{k, n} = \phi(\Lambda, \alpha^{\xi,\xi}_{k, n}) + \frac{1}{\sqrt{k}}\phi(\mathcal{E}_{k,n}, \alpha^{\xi,\xi}_{k, n}).
\end{align*} Due to  \cite[Theorem 10.8]{Kosorok2008}, we have $\phi(\Lambda,\alpha^{\xi,\xi}_{k, n}) \weak_\xi \phi(\Lambda, \mathbb{G}_{\widehat{\Lambda}}) $. Moreover, as $n,k, n/k \to \infty$, \begin{align*}
    \mathbb{E}_{\xi} \left|\frac{1}{\sqrt{k}}\phi(\mathcal{E}_{k, n}, \alpha^{\xi,\xi}_{k, n}) + B^{\xi, \xi}_{k, n}\right| \stp 0.
\end{align*} Now, using \Cref{lemma:condweaklemma}, we have as $n,k, n/k \to \infty$,
\begin{align*}
    \frac{\mu}{\tau} \sqrt{k} (\Delta^{\xi, \xi}_{k, n} - \Delta_{k, n}) \weak_\xi \phi(\Lambda, \mathbb{G}_{\widehat{\Lambda}}) = 6 \int_0^1 \left[\Lambda(u, 1)\mathbb{G}_{\widehat{\Lambda}}(u, 1) - \Lambda(1, u)\mathbb{G}_{\widehat{\Lambda}}(1, u) \right]\mathrm{d}u.
\end{align*}
\end{proof}

Finally, we provide an additional result where the bootstrapping estimator $\eta^{\xi, \xi}_{k, n} (X|Y)$ is represented using bootstrapping ranks. This is helpful for computational purposes. Consider i.i.d. data $(X_1,Y_1), \ldots, (X_n,Y_n)$ and the independent i.i.d. sample $\xi_1,\ldots, \xi_n$ (used for multiplier bootstrap).
 Let $Y_{(1)} \geq ... \geq Y_{(n)}$ be the decreasing order statistics of $\{Y_i, 1\le i \le n\}$ with $X_{[i]}, \xi_{[i]}$ being the $X$, $\xi$ concomitant variables associated with $Y_{(i)}$ for $i=1,\ldots,n$.  Now denote for $i=1,\ldots, n$ the following: \begin{align*}
        R^{\xi}(X_i) &:= n\ov{F^{\xi}_{n, 1}}(X_i), \\
        R^{\xi}(Y_i) &:= n\ov{F^{\xi}_{n, 2}}(Y_i),\\
        \tau(k) &:= \max\{i = 1,...,n : R^{\xi}(Y_{(i)}) < k\}.
    \end{align*}
\begin{lemma}\label{lem:bootstrappingestimator}
    The following identity holds:
    \begin{align*}
        \eta^{\xi, \xi}_{k, n} (X|Y) &= \frac{3}{k^3} \sum_{i \leq \tau(k)} \sum_{j \leq \tau(k)} \frac{\xi_{[i]}\xi_{[j]}}{\overline{\xi}_n^2} \left(k - \max\left(R^{\xi}(X_{[i]}), R^{\xi}(X_{[j]})\right)\right)_{+}.
    \end{align*}
\end{lemma}

\begin{proof}[Proof of lemma \ref{lem:bootstrappingestimator}]
Since $\ov{F}^{\xi}_{n, 1}(X_i) = {R^{\xi}(X_i)}/{n}$ and $\ov{F}^{\xi}_{n, 2}(Y_i) =  {R^{\xi}(Y_i)}/{n}$, we have
\begin{align*}
    \ov{C}^{\xi, \xi}_n(x, y) &= \frac{1}{n} \sum_{i=1}^n \frac{\xi_i}{\overline{\xi}_n}\bI(\ov{F}^{\xi}_{n, 1}(X_i) < x, \ov{F}^{\xi}_{n, 2}(Y_i) < y)\\
&= \frac{1}{n} \sum_{i=1}^n \frac{\xi_i}{\overline{\xi}_n} \bI(R^{\xi}(X_i) < nx, R^{\xi}(Y_i) < ny).
\end{align*} Therefore
\begin{align}
    \Lambda^{\xi, \xi}_{k, n}(x, 1) &= \frac{1}{k} \sum_{i=1}^n \frac{\xi_i}{\overline{\xi}_n} \bI(R^{\xi}(X_i) < kx, R^{\xi}(Y_i) < k) = \frac{1}{k} \sum_{i=1}^{\tau(k)} \frac{\xi_{[i]}}{\overline{\xi}_n}  \bI(R^{\xi}(X_{[i]}) < kx). \label{eq:lambdaemp}
\end{align} Substitute the expression for $\Lambda^{\xi, \xi}_{k, n}(x, 1)$ from \eqref{eq:lambdaemp} in the definition of $\eta^{\xi, \xi}_{k, n}(X|Y)$ we get
\begin{align*}
    \eta^{\xi, \xi}_{k, n}(X|Y) &= \frac{3}{k^2} \int_0^1 \sum_{i \leq \tau(k)} \sum_{j \leq \tau(k)} \frac{\xi_{[i]}\xi_{[j]}}{\overline{\xi}_n^2}\bI(R^{\xi}(X_{[i]}) < kx, R^{\xi}(X_{[j]}) < kx) \mathrm{d}x \\
    &= \frac{3}{k^3}  \sum_{i \leq \tau(k)} \sum_{j \leq \tau(k)} \frac{\xi_{[i]}\xi_{[j]}}{\overline{\xi}_n^2}\left(k - \max\left(R^{\xi}(X_{[i]}), R^{\xi}(X_{[j]})\right)\right)_+
\end{align*} 
proving the claimed identity.
\end{proof}

\newpage

\section{Cryptocurrency data analysis -- additional plots}\label{app:datacrypto}

In this section, we have consolidated some of the expository plots from the Cryptocurrency data analysis conducted in \Cref{subsec:datacrypto}. For notation, refer to the text in \Cref{subsec:datacrypto}.

\begin{figure}[h]
    \centering
    \includegraphics[width=0.85\textwidth]{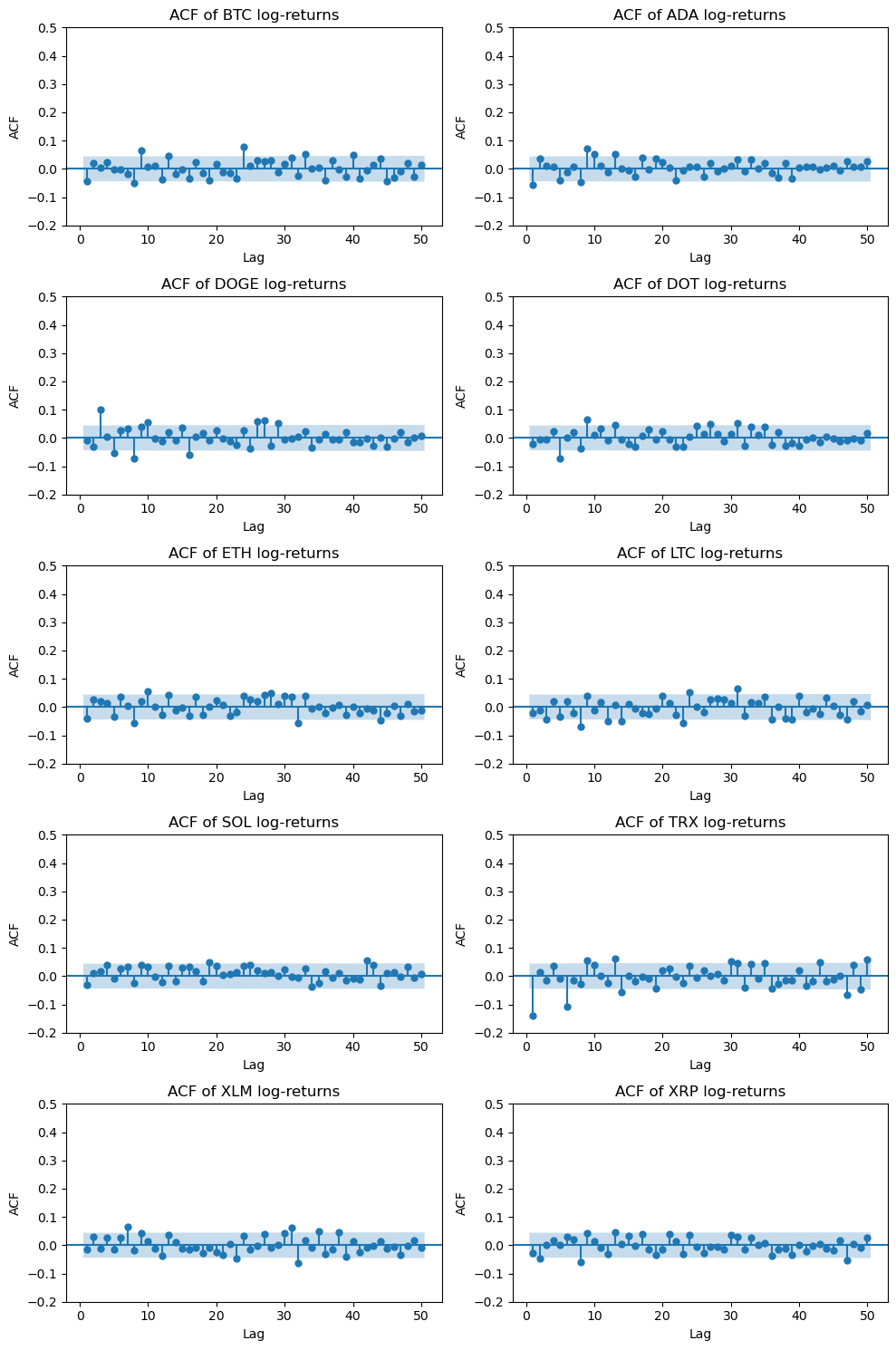} 
    \caption{Cryptocurrency data from \Cref{subsec:datacrypto}: 
        autocorrelation plots for the log-returns of closing prices from 1-1-2021 to 31-12-2025 for the ten major cryptocurrencies considered, showing the sample autocorrelation up to lag 50. The absence of significant autocorrelation suggests that the log-returns series can be reasonably modeled as independent and identically distributed (i.i.d.).
    }
    \label{fig:acf_all}
\end{figure}

Next, we have nine separate graphs in each of Figures \ref{fig:close_pos_boot}-\ref{fig:etaYX_boot_neg}. Each graph uses a bivariate data set $(X,Y^C)$ or $(X^{-},Y^{-C})$; therefore, one of the variables is $X$ (or, $X^-$) representing the log-return (or, negative log-return, respectively) of the closing price of Bitcoin; the other variable is $Y^C$ (or $Y^{-C}$) representing one of the other nine cryptocurrencies where we change $\text{C}\in \{\text{ADA, DOGE, DOT, ETH, LTC, SOL, TRX, XLM, XRP}\}$ to get the nine different graphs. The graphs give 95\% confidence bound around the estimators $\Delta_{k,n}(X,Y^C), \Delta_{k,n}(X^-,Y^{-C}),\eta_{k,n}(X|Y^C), \eta_{k,n}(X^-|Y^{-C}), \eta_{k,n}(Y^C|C)$, and\newline $\eta_{k,n}(Y^{-C}|X^{-C})$ plotted for $k\in \mathcal{K}=\{100,110,\ldots,500\}$.

\begin{figure}[b]
    \centering
    \includegraphics[width=\textwidth]{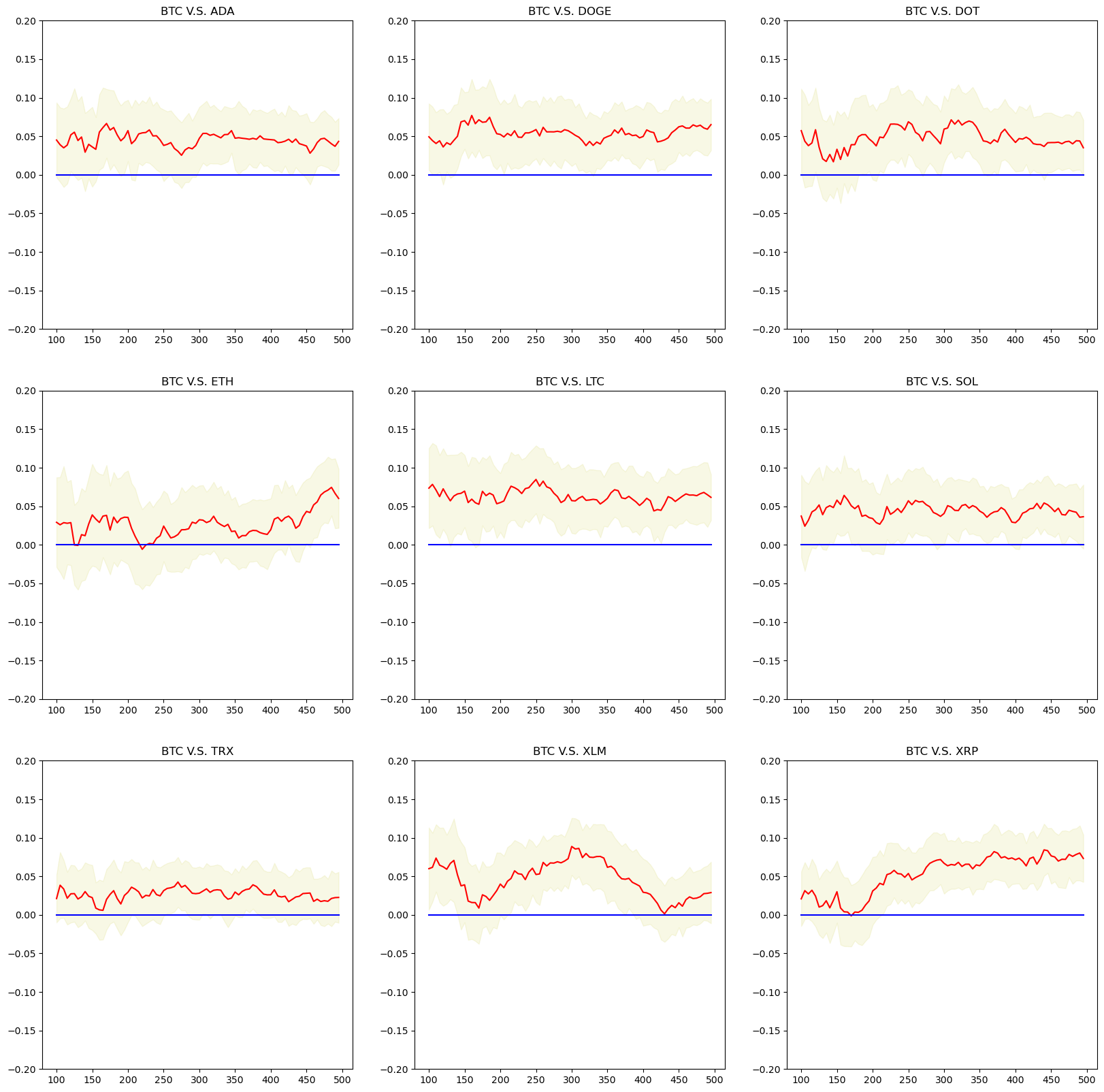} 
    \caption{ 
        Cryptocurrency data from \Cref{subsec:datacrypto}:  The red trajectory is the plot of the estimate $\Delta_{k,n}(X,Y^C)$  for the tail asymmetry measure  $\Delta(X,Y^C)$ of \textbf{positive} log-returns with different choices of $k\in \mathcal{K}$. The yellow band represents a 95\% confidence interval around $\Delta_{k,n}(X,Y^C)$ constructed using $B=100$ multiplier bootstrap samples. The blue horizontal line indicates the zero baseline corresponding to the null hypothesis of $\mathcal{H}_0:\Delta(X,Y^C)=0$. }
    \label{fig:close_pos_boot}
\end{figure}

\begin{figure}[t]
    \centering
    \includegraphics[width=\textwidth]{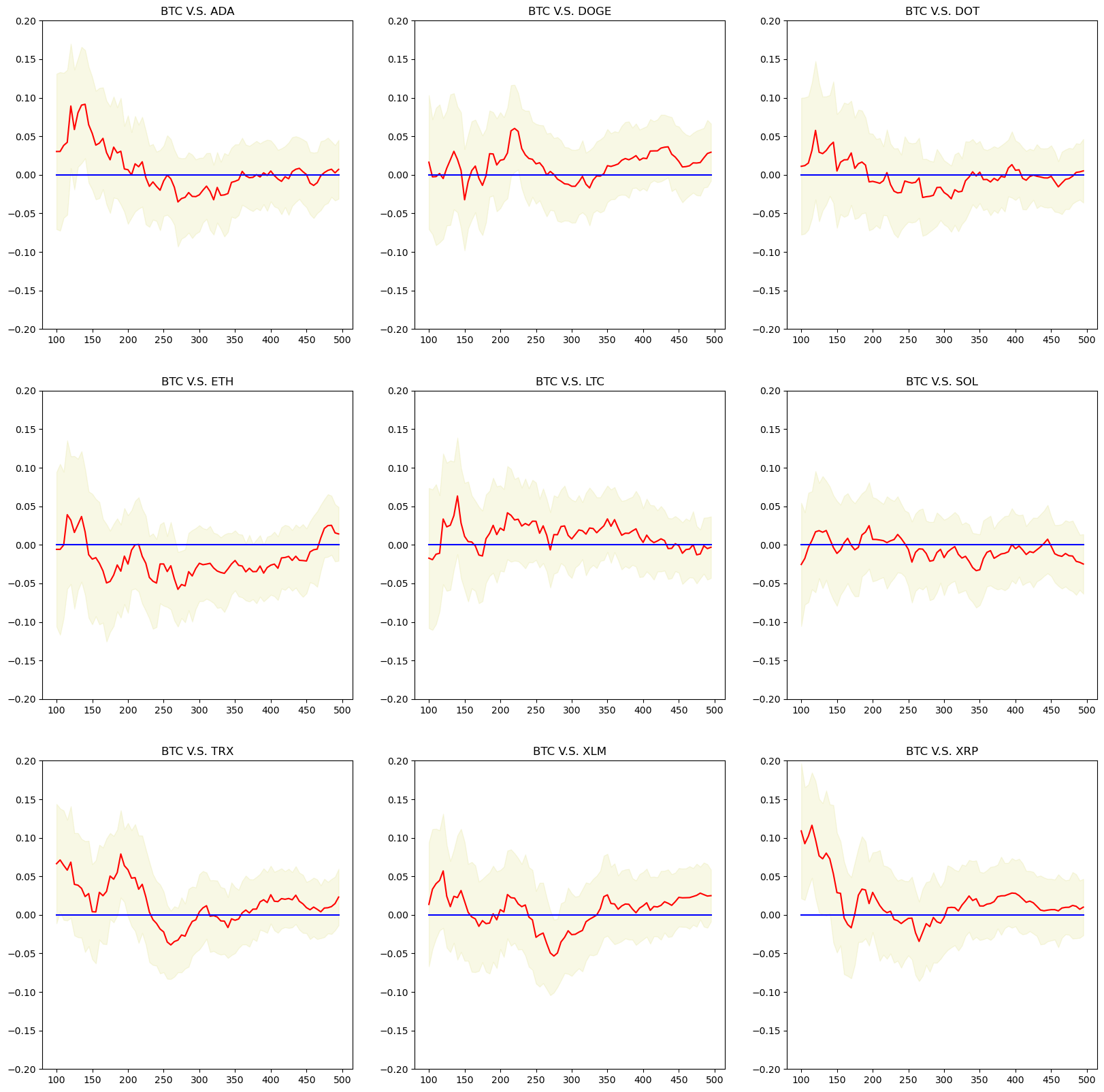} 
    \caption{ 
        Cryptocurrency data from \Cref{subsec:datacrypto}:  The red trajectory is the plot of the estimate $\Delta_{k,n}(X^-,Y^{-C})$  for the tail asymmetry measure  $\Delta(X^-,Y^{-C})$ of \textbf{negative} log-returns with different choices of $k\in \mathcal{K}$. The yellow band represents a 95\% confidence interval around $\Delta_{k,n}(X^-,Y^{-C})$ constructed using $B=100$ multiplier bootstrap samples. The blue horizontal line indicates the zero baseline corresponding to the null hypothesis of $\mathcal{H}_0:\Delta(X^-,Y^{-C})=0$. }
    \label{fig:close_neg_boot}
\end{figure}

\begin{figure}[t]
    \centering
    \includegraphics[width=\textwidth]{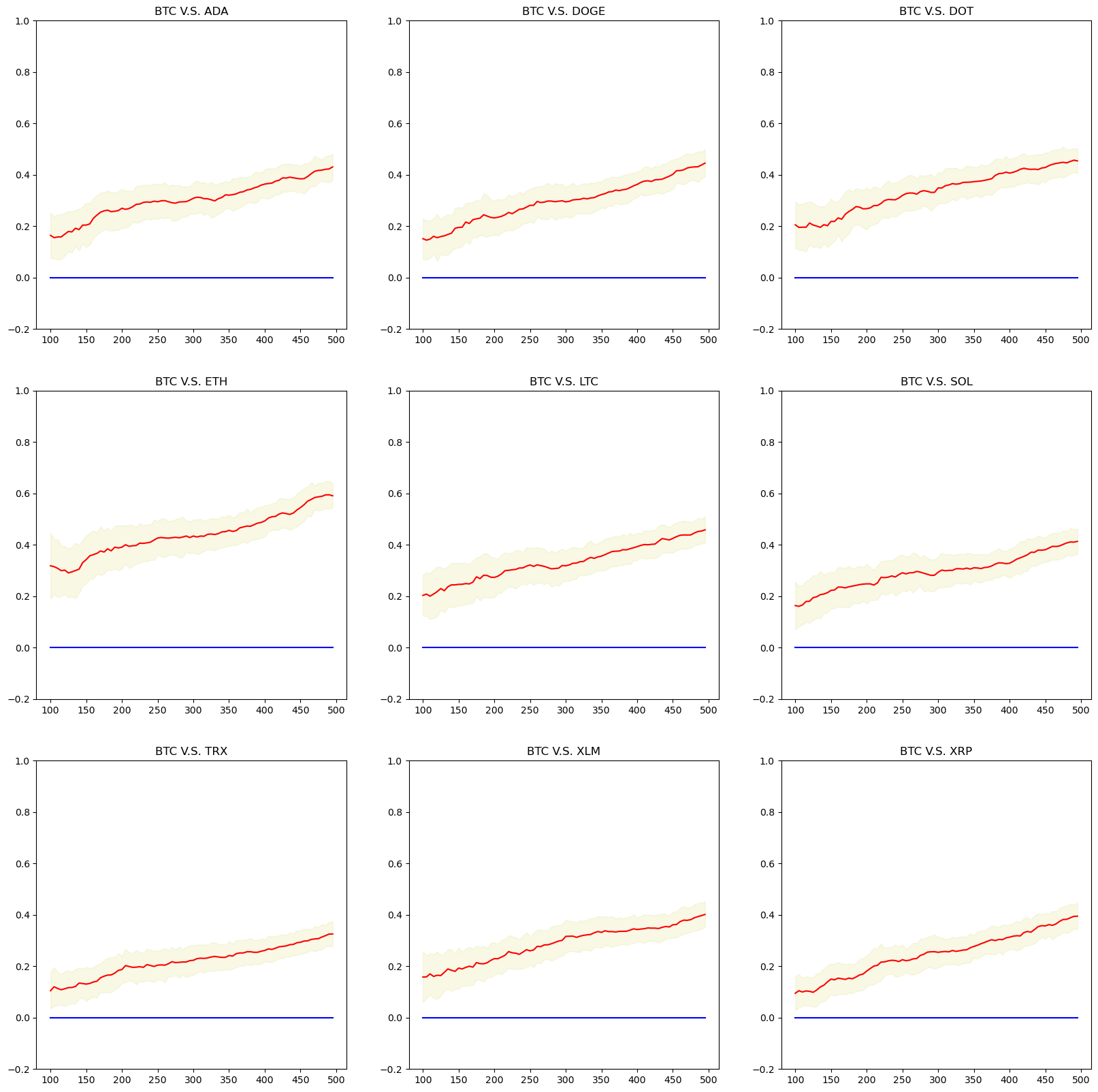} 
    \caption{ 
     Cryptocurrency data from \Cref{subsec:datacrypto}:  The red trajectory is the plot of the estimate $\eta_{k,n}(X|Y^{C})$  for the ETA measure  $\eta(X|Y^{C})$ of \textbf{positive} log-returns with different choices of $k\in \mathcal{K}$. The yellow band represents a 95\% confidence interval around $\eta_{k,n}(X|Y^{C})$ constructed using $B=100$ multiplier bootstrap samples. The blue horizontal line indicates the zero baseline corresponding to the null hypothesis of $\mathcal{H}_0:\eta(X|Y^{C})=0$. }
    \label{fig:etaXY_boot_pos}
\end{figure}

\begin{figure}[t]
    \centering
    \includegraphics[width=\textwidth]{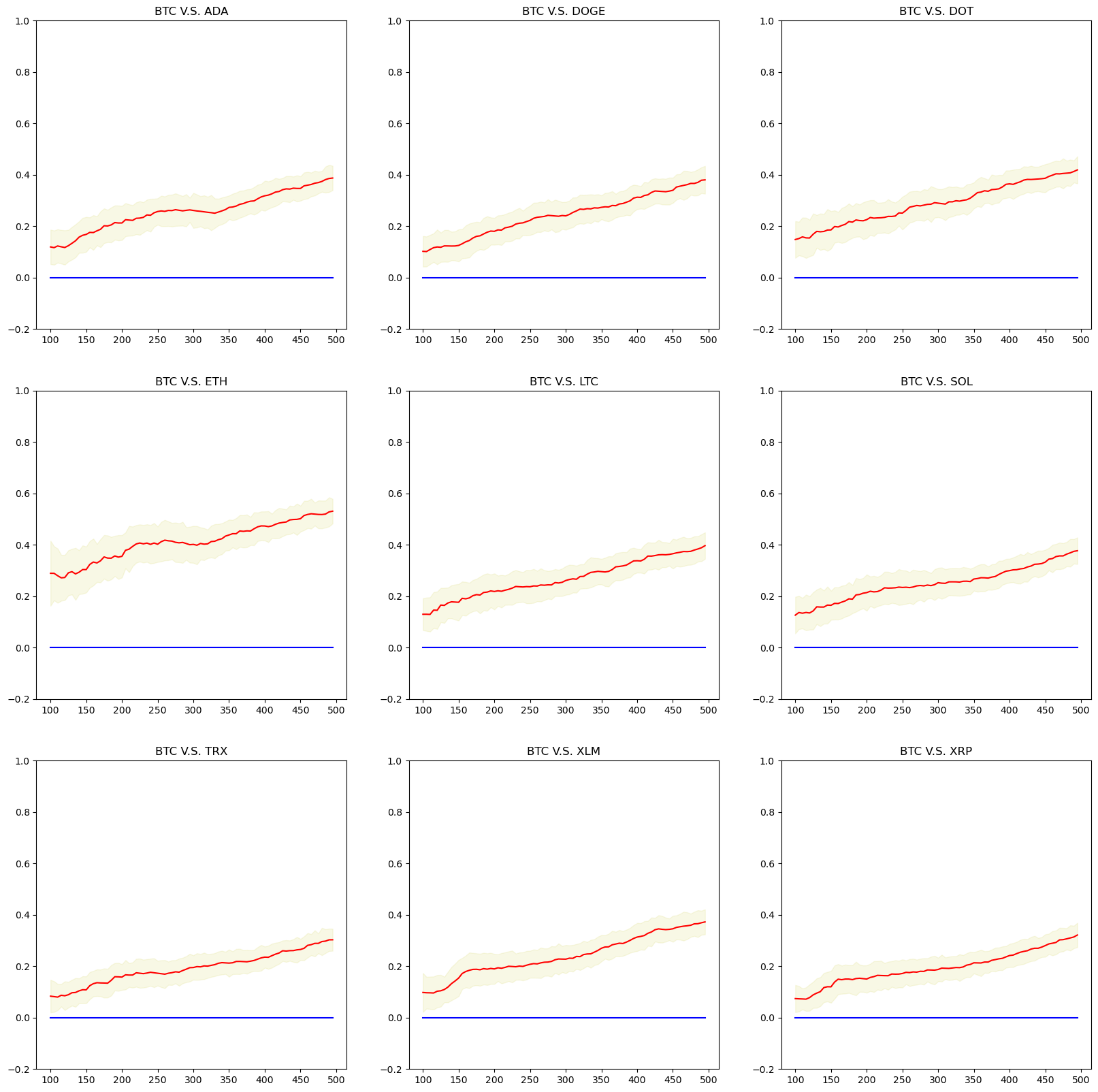} 
\caption{ 
     Cryptocurrency data from \Cref{subsec:datacrypto}:  The red trajectory is the plot of the estimate $\eta_{k,n}(Y^{C}|X)$  for the ETA measure  $\eta(Y^{C}|X)$ of \textbf{positive} log-returns with different choices of $k\in \mathcal{K}$. The yellow band represents a 95\% confidence interval around $\eta_{k,n}(Y^{C}|X)$ constructed using $B=100$ multiplier bootstrap samples. The blue horizontal line indicates the zero baseline corresponding to the null hypothesis of $\mathcal{H}_0:\eta(Y^{C}|X)=0$. }
    \label{fig:etaYX_boot_pos}
\end{figure}

\begin{figure}[t]
    \centering
    \includegraphics[width=\textwidth]{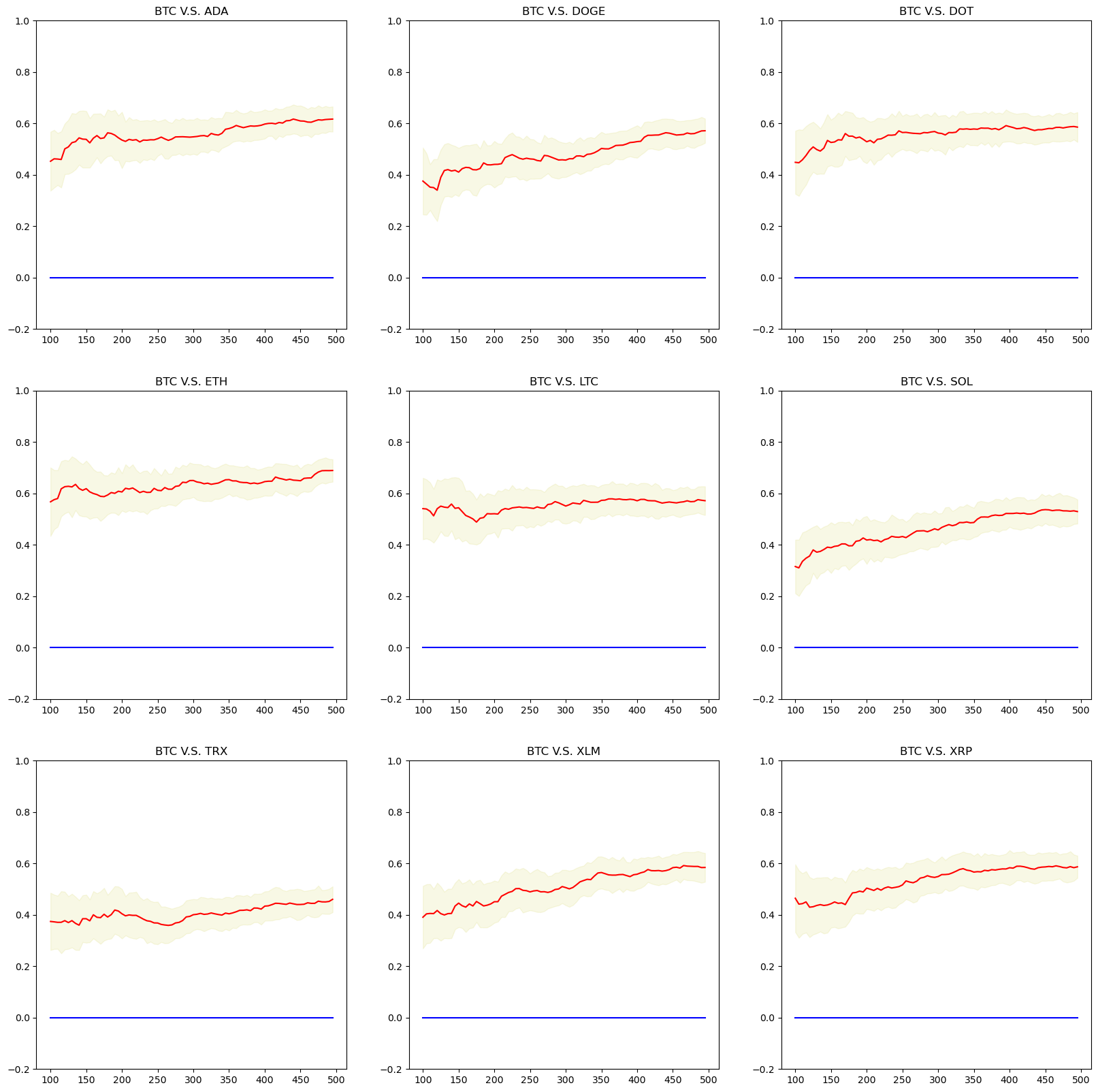} 
    \caption{ 
     Cryptocurrency data from \Cref{subsec:datacrypto}:  The red trajectory is the plot of the estimate $\eta_{k,n}(X^-|Y^{-C})$  for the ETA measure  $\eta(X^-|Y^{-C})$ of \textbf{negative} log-returns with different choices of $k\in \mathcal{K}$. The yellow band represents a 95\% confidence interval around $\eta_{k,n}(X^-|Y^{-C})$ constructed using $B=100$ multiplier bootstrap samples. The blue horizontal line indicates the zero baseline corresponding to the null hypothesis of $\mathcal{H}_0:\eta(X^-|Y^{-C})=0$. }
    \label{fig:etaXY_boot_neg}
\end{figure}

\begin{figure}[t]
    \centering
    \includegraphics[width=\textwidth]{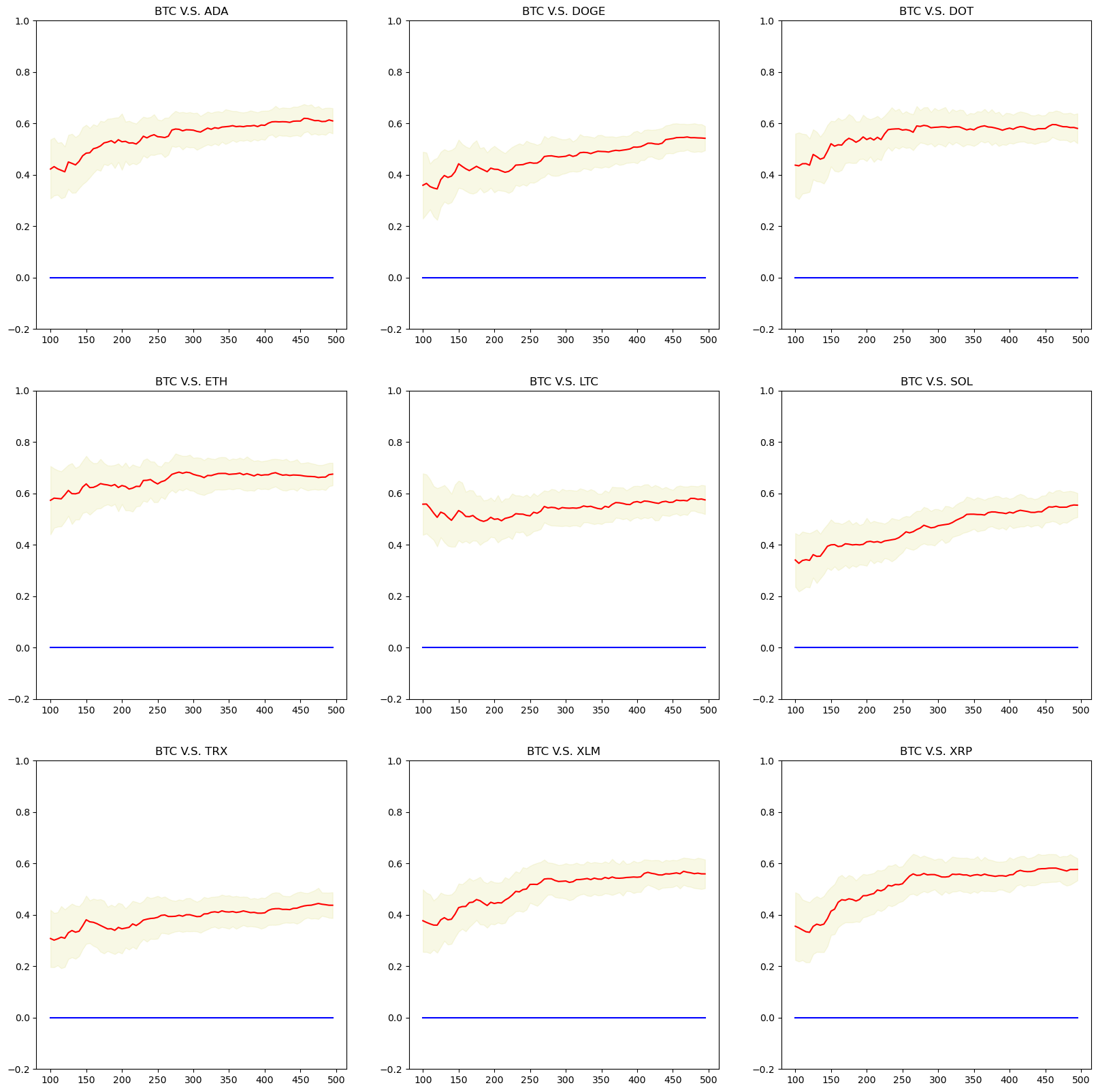} 
\caption{ 
     Cryptocurrency data from \Cref{subsec:datacrypto}:  The red trajectory is the plot of the estimate $\eta_{k,n}(Y^{-C}|X^-)$  for the ETA measure  $\eta(Y^{-C}|X^-)$ of \textbf{negative} log-returns with different choices of $k\in \mathcal{K}$. The yellow band represents a 95\% confidence interval around $\eta_{k,n}(Y^{-C}|X^-)$ constructed using $B=100$ multiplier bootstrap samples. The blue horizontal line indicates the zero baseline corresponding to the null hypothesis of $\mathcal{H}_0:\eta(Y^{-C}|X^-)=0$. }
    \label{fig:etaYX_boot_neg}
\end{figure}

\end{document}